\documentclass[a4paper, notitlepage, 11pt]{article}
\usepackage{geometry}
\fontfamily{times}
\usepackage{setspace,relsize}               
\usepackage{moreverb}                        
\usepackage{url}
\usepackage{hyperref}
\hypersetup{colorlinks=true,citecolor=blue}
\usepackage{subfigure}
\usepackage{enumitem}
\usepackage{amsmath}
\usepackage{mathtools} 
\usepackage{amsthm}
\usepackage{amsmath}
\usepackage{amssymb}
\usepackage{indentfirst}
\usepackage{todonotes}
\usepackage[authoryear,round]{natbib}
\usepackage[pdftex]{lscape}
\usepackage[utf8]{inputenc}
\usepackage{multirow}
\usepackage{chngcntr}

\title{\vspace{-9ex}\centering \bf On
the normalized power prior}
\author{
Luiz Max de Carvalho$^1$, Joseph G. Ibrahim$^2$\\
1 - School of Applied Mathematics (EMAp), Get\'ulio Vargas Foundation (FGV).\\
2 - Department of Biostatistics, University of North Carolina at Chapel Hill.\\
}

\newtheorem{theorem}{Theorem}[]
\newtheorem{lemma}{Lemma}[]
\newtheorem{proposition}{Proposition}[]
\newtheorem{remark}{Remark}[]
\begin{document}
\maketitle

\begin{abstract}
The power prior is a popular tool for constructing informative prior distributions based on historical data.
The method consists of raising the likelihood to a discounting factor in order to control the amount of information borrowed from the historical data.
It is customary to perform a sensitivity analysis reporting results for a range of values of the discounting factor.
However, one often wishes to assign it a prior distribution and estimate it jointly with the parameters, which in turn necessitates the computation of a normalising constant.
In this paper we are concerned with how to recycle computations from a sensitivity analysis in order to approximately sample from joint posterior of the parameters and the discounting factor.
We first show a few important properties of the normalising constant and then use these results to motivate a bisection-type algorithm for computing it on a fixed budget of evaluations.
We give a large array of illustrations and discuss cases where the normalising constant is known in closed-form and where it is not.
We show that the proposed method produces approximate posteriors that are very close to the exact distributions when those are available and also produces posteriors that cover the data-generating parameters with higher probability in the intractable case.
Our results show that proper inclusion the normalising constant is crucial to the correct quantification of uncertainty and that the proposed method is an accurate and easy to implement technique to include this normalisation, being applicable to a large class of models. 

\end{abstract}

\section{Background}

Power priors~\citep{Ibrahim2000,Ibrahim2015} are are popular method of constructing informative priors, and are widely used in fields such medical research, where elicitation of informative priors is crucial.
When historical data are available, power priors allow the elicitation of informative priors by borrowing information from the historical data.
This is accomplished by raising the likelihood of the historical data to a scalar discounting factor $a_0$, usually taken to be $0 \leq a_0 \leq 1$.
When $a_0 = 0$ the historical data receives no weight and thus no information is borrowed, whereas $a_0 = 1$ represents full borrowing of the information contained in the historical data to inform the prior.
In many settings, this construction of an informative prior can be shown to be optimal in an information-processing sense~\citep{Zellner2002,Ibrahim2003}.

To make the presentation more precise, let the observed data be $D_0 = \{d_{01}, d_{02}, \ldots, d_{0N_0} \}$, $d_{0i} \in \mathcal{X} \subseteq \mathbb{R}^d$ and let $L(D_0 \mid \theta)$ be a likelihood function assumed to be finite for all arguments $\theta \in \boldsymbol\Theta \subseteq \mathbb{R}^q$. 
The simplest formulation of the power prior reads
\begin{equation}
 \label{eq:power_prior_simple}
 \pi(\theta \mid D_0, a_0) \propto L(D_0 \mid \theta)^{a_0}\pi(\theta),
\end{equation}
where  $\pi$ is called the~\textit{initial} prior and $a_0$ is a scalar, usually taken to be in $[0, 1]$.
The scalar $a_0$ controls the amount of information from historical data that is included in the analysis of the current data~\citep{Ibrahim2000}.
A commonly adopted practice is to fix $a_0$ and assess the sensitivity of results to different values, including $a_0 = 0$ (no borrowing) and $a_0 = 1$ (full borrowing, see~\cite{Ibrahim2015}, Section 5).
One might also be interested in accommodating uncertainty about the relative weighting of the historical data by placing a prior $\pi_A$ on $a_0$.
This leads to what we will call the~\textit{unnormalised} power prior
\begin{equation}
 \label{eq:unnormalised_power_prior}
 \pi(\theta, a_0 \mid D_0, \delta) = \pi(\theta \mid D_0, a_0)\pi_A(a_0 \mid \delta) \propto L(D_0 \mid \theta)^{a_0}\pi(\theta)\pi_A(a_0).
\end{equation}
As observed by~\cite{Neuenschwander2009}, this formulation does not lead to a correct joint posterior distribution for $(\theta, a_0)$ because the normalsing constant of~(\ref{eq:power_prior_simple}),
\begin{equation}
 \label{eq:normalising_constant}
 c(a_0) : = \int_{\Theta} L( D_0 \mid \theta)^{a_0}\pi(\theta) \, d\theta,
\end{equation}
is missing.
The~\textit{normalised} power prior is defined as~\citep{Duan2006a,Duan2006b}:
\begin{equation}
 \label{eq:normalised_power_prior}
  \pi(\theta, a_0 \mid D_0, \delta) = \frac{L(D_0| \theta)^{a_0} \pi(\theta)  \pi_A(a_0\mid \delta)}{c(a_0)}.
\end{equation}

In light of new (current) data $D = \{d_1, d_2, \ldots, d_N\}$, we thus have the joint posterior
\begin{equation}
 \label{eq:joint_posterior}
 p(\theta, a_0 \mid D_0, D, \delta) \propto \frac{1}{c(a_0)} L(D \mid \theta) L(D_0 \mid \theta)^{a_0} \pi(\theta)\pi_A(a_0 \mid \delta).
\end{equation}
In this setting, $a_0$ becomes a parameter we are interested in learning about in light of the data and thus we arrive at the marginal posterior
\begin{align}
 \nonumber
   p(a_0 \mid D_0, D, \delta)  &= \int_{\Theta}  p(\theta, a_0 \mid D_0, D, \delta) \, d\theta,\\
  \label{eq:marginal_posterior_a0}
  & \propto \frac{\pi_A(a_0\mid \delta)}{c(a_0)} \int_{\Theta}  L(D_0| \theta)^{a_0} \pi(\theta) L(D | \theta) \, d\theta,
\end{align}
which involves the computation of not one but two potentially high-dimensional integrals.

Unfortunately, the posterior distribution in~(\ref{eq:joint_posterior}) is in the class of so-called doubly-intractable distributions and its exact computation depends on advanced Markov chain Monte Carlo (MCMC) techniques.
In this paper we study the theoretical properties of the normalising constant and use our findings to guide informed designs for sensitivity analysis.
Further, we explore a simple way to recycle computations from a sensitivity analysis in order to sample from an approximate joint posterior of $a_0$ and $\theta$.

The paper is organised as follows: in Section~\ref{sec:theory} we present a few results on the propriety of the power prior and the properties of the normalising constant, $c(a_0)$.
We give general results as well as specific formulae for the exponential family of probability distributions and its conjugate prior.
Section~\ref{sec:computation} discusses the computational aspects of approximating $c(a_0)$ when it is not known in closed-form and Section~\ref{sec:illustrations} brings a large array of illustrations of the normalised power prior in examples where the normalising constant is known in closed-form and situations where it is not.
We conclude with a discussion of the findings and future directions in Section~\ref{sec:discussion}.

\section{Theory}
\label{sec:theory}

We begin by describing a few results on the properties of the power prior and its normalised version.
First, we show that the normalised power prior is always well-defined when the initial prior is proper in Theorem~\ref{thm:integrability}, for which we give an elementary proof in Appendix~\ref{sec:further_proofs}.

\begin{theorem}
\label{thm:integrability}
 Assume $\int_{\mathcal{X}} L( x \mid \theta)\, dx < \infty$.
 In addition, assume $\pi$ is proper, i.e., $\int_{\boldsymbol\Theta} \pi(\theta) \, d\theta = 1$.
 Then, $c(a_0) = \int_{\Theta}  L( D_0 \mid \theta)^{a_0}\pi(\theta) \, d\theta <\infty$ for $a_0 \geq 0$.
\end{theorem}
\begin{proof}\let\qed\relax
See Appendix~\ref{sec:further_proofs}.
\end{proof}

Theorem~\ref{thm:integrability} thus shows that the expression in~(\ref{eq:normalised_power_prior}) leads to a valid joint prior on $(\theta, a_0)$.
While scientific interest usually lies with $a_0 \in [0, 1]$, showing the result holds also for $a_0 > 1$ might find use in other fields, such as the analysis of complex surveys, where the likelihood is raised to a power that is inversely proportional to a selection probability~\citep{Savitsky2016}.
In many applications one usually has a collection of historical data sets, with different sample sizes and particular (relative) reliabilities that one would like to include in a power prior analysis.
Remark~\ref{rmk:historical} extends Theorem~\ref{thm:integrability} for the situation where multiple historical data sets are available and the analyst desires to include them simultaneously, each with a weight $a_{0k}$.

\begin{remark}
\label{rmk:historical}
 The power prior on multiple (independent) historical data sets is also a proper density.
\end{remark}
\begin{proof}\let\qed\relax
See Appendix~\ref{sec:further_proofs}.
\end{proof}

These two results give solid footing to the normalised power prior as well as tempered likelihood techniques, for which propriety is usually assumed without proof or proved only for specific cases~\citep{Duan2006a, Savitsky2016}.

\subsection{Properties of the normalising constant $c(a_0)$}
\label{sec:properties}

In order to aid computation, it is useful to study some of the properties of the normalising constant, $c(a_0)$, seen as a function of the scalar $a_0$. 
First, we show that  $c(a_0)$ is strictly convex (Lemma~\ref{lm:convex_norm_constant}), which motivates specific algorithms for its approximation.

\begin{lemma}
\label{lm:convex_norm_constant}
Assume $L(D_0 \mid \theta)$ is continuous with respect to $\theta$.
Then the normalising constant is a strictly convex function of $a_0$.
\end{lemma}
\begin{proof}\let\qed\relax
See Appendix~\ref{sec:further_proofs}.
\end{proof}

For the goals of this paper, it would be useful to know more about the shape of $c(a_0)$, more specifically if and when its derivatives change signs.
For computational stability reasons, one is usually interested in $l(a_0) := \log(c(a_0)) $ instead of $c(a_0)$ and hence it is also useful to study the derivative of the log-normalising constant, $l^\prime(a_0) = c^\prime(a_0)/c(a_0)$.
A key observation is that $l^\prime(a_0)$ changes signs at the same point as $c^\prime(a_0)$ does, a feature that can be exploited when designing algorithms for approximating $l(a_0)$ (see Section~\ref{sec:efficient_computation_ca0}).

Next, we state Remark~\ref{rmk:discrete_decreasing}, that shows that for the large class of discrete likelihoods (Bernoulli, Poisson, etc), $c(a_0)$ is monotonic.
\begin{remark}
 \label{rmk:discrete_decreasing}
 When $L(D \mid \theta)$ is a discrete likelihood, $c(a_0)$ is monotonically decreasing in $a_0$.
\end{remark}
\begin{proof}
 The proof is immediate from Lemma~\ref{lm:convex_norm_constant} and the fact that for a non-degenerate discrete likelihood the function $\log(L(D\mid\theta))$ is strictly negative and hence so is its expectation under the power prior (see Proposition~\ref{prop:c_is_Cinfinity} in Appendix~\ref{sec:further_proofs}).
\end{proof}
This will find application in the adaptive grid building described in Section~\ref{sec:adapt_grid}.

\subsubsection{Exponential family}
\label{sec:expo_family}

A large class of models routinely employed in applications is the exponential family of distributions, which includes the Gaussian and Gamma families, as well as the class of generalised linear models~\citep{Nelder1972, Mccullagh1989}.
Here we give expressions for the normalising constant when the likelihood is in the exponential family.
Furthermore, we also derive the marginal posterior of $a_0$ when the initial prior $\pi(\theta)$ is in the conjugate class~\citep{Diaconis1979}. 

Suppose $L(D_0 \mid \theta)$ is in the exponential family:
\begin{equation*}
 L(D_0 \mid \theta) = \boldsymbol h(D_0) \exp \left( \eta(\theta)^T \left(\boldsymbol S(D_0)) \right) - N_0 A(\theta) \right),
\end{equation*}
where $\boldsymbol h(D_0) := \prod_{i = 1}^{N_0} h(d_{0i})$ and $\boldsymbol S(D_0) :=  \sum_{i=1}^{N_0} T(d_{0i})$.
Thus we have
\begin{align}
\nonumber
 c(a_0) &=  \int_{\boldsymbol\Theta} \left[ \boldsymbol h(D_0) \exp \left( \eta(\theta)^T \boldsymbol S(D_0) - N_0 A(\theta) \right) \right]^{a_0}\pi(\theta) \, d\theta, \\
  \label{eq:expo_family_const}
 &= \boldsymbol h(D_0)^{a_0}\int_{\boldsymbol\Theta} \exp \left( \eta(\theta)^T a_0 \boldsymbol S(D_0) \right) \exp\left(- a_0N_0 A(\theta) \right) \pi(\theta) \, d\theta.
\end{align}
The derivative (see Proposition~\ref{prop:c_is_Cinfinity}) evaluates to 
\begin{equation}
\label{eq:expo_family_deriv_general}
 c^\prime(a_0) = \log(\boldsymbol h(D_0)) +  \int_{\boldsymbol\Theta} \left[ \eta(\theta)^T a_0 \boldsymbol S(D_0) \right]  f_{a_0}(D_0; \theta) \, d\theta -  a_0N_0\int_{\boldsymbol\Theta} f_{a_0}(D_0; \theta)  A(\theta) \, d\theta,
\end{equation}
where $f_{a_0}(D_0; \theta) := L(D_0 \mid \theta)^{a_0}\pi(\theta)$.

These results can be refined further if we restrict the class of initial priors.
If we choose $\pi(\theta)$ to be conjugate to $L(D_0 \mid \theta)$~\citep{Diaconis1979}, i.e.
\begin{equation*}
 \label{eq:conj_exp_family}
 \pi(\theta \mid \tau, n_0) = H(\tau, n_0) \exp\{ \tau^T\eta(\theta) - n_0A(\theta) \},
\end{equation*}
we have 
\begin{align}
\nonumber
 c(a_0) &= \boldsymbol h(D_0)^{a_0} H(\tau, n_0) \int_{\boldsymbol\Theta}  \exp \left[ \eta(\theta)^T \left( \tau + a_0\boldsymbol S(D_0) \right) -(n_0  + a_0N_0) A(\theta) \right] \, d\theta, \\
  \label{eq:expo_family_const_conj}
 &= \frac{\boldsymbol h(D_0)^{a_0} H(\tau, n_0)}{H\left( \left[\tau + a_0\boldsymbol S(D_0)\right]^T, n_0  + a_0N_0 \right)}.
\end{align}
Following~(\ref{eq:marginal_posterior_a0}), the marginal posterior for $a_0$ is 
\begin{equation}
 \label{eq:marginal_posterior_a0_expoFamily}
 p(a_0 \mid D_0, D, \delta) \propto \frac{H\left( \left[\tau + a_0\boldsymbol S(D_0)\right]^T, n_0  + a_0N_0 \right)}{H\left( \left[\tau + a_0\boldsymbol S(D_0) + \boldsymbol S(D) \right]^T, n_0  + a_0N_0 + N \right)} \frac{\boldsymbol h(D)}{\boldsymbol h(D_0)^{a_0} } \pi_A(a_0 \mid \delta).
\end{equation}

\section{Computation}
\label{sec:computation}

In this section we propose a way to approximate $c(a_0)$ at a grid of values, while simultaneously picking the grid values themselves.

\subsection{Efficiently approximating $c(a_0)$}
\label{sec:efficient_computation_ca0}

While the exponential family is a broad class of models, for many models of practical interest $c(a_0)$ is not known in closed form, and hence must be computed approximately.
As discussed above (and by~\cite{Neuenschwander2009}), it is important to include $c(a_0)$ in the calculation of the posterior when $a_0$ is allowed to vary and assigned its own prior.
An example where this would be important is when we need to normalise the power prior for use within a Markov chain Monte Carlo (MCMC) procedure.
If one wants to avoid developing a customised MCMC sampler for this situation (see Section~\ref{sec:discussion}), one needs a simple yet accurate way of approximating $c(a_0)$ and its logarithm, $l(a_0)$.

Here we take the following approach to approximating $c(a_0)$: first, define a grid of values $\boldsymbol a^{\text{est}} = \{ a^{\text{est}}_1, \ldots, a^{\text{est}}_J \}$ for a typically modest grid size $J$.
Using a marginal likelihood approximation method (see below), compute an estimate of $c(a_0)$ for each point in $\boldsymbol a^{\text{est}}$, obtaining a set of estimates.
Consider an approximating function $g_{\boldsymbol\xi} : [0, \infty) \to (0, \infty)$, indexed by a set of parameters $\boldsymbol\xi$.
For instance, $g_{\boldsymbol\xi}$ could be a linear model, a generalised additive model (GAM) or a Gaussian process.
We can then use  $\boldsymbol a^{\text{est}}$ and $\hat{\boldsymbol c}(\boldsymbol a^{\text{est}})$ as data to learn about $\boldsymbol\xi$.
Once we obtain an estimate  $\hat{\boldsymbol \xi}$, $c(\cdot)$ can be approximated at any point $z$ by the prediction $g_{\hat{\boldsymbol \xi}}(z)$.

In order to simplify implementation, in our applications we found it useful to create a grid of size $K \gg J$, $\boldsymbol a^{\text{pred}} = \{ a^{\text{pred}}_1, \ldots, a^{\text{pred}}_K \}$, and then compute the predictions $\boldsymbol g^{\text{pred}} := g_{\hat{\boldsymbol \xi}}( \boldsymbol a^{\text{pred}} )$.
We can then use this dictionary of values to obtain an approximate value of $c(a_0)$ by simple interpolation.
This approach allows one to evaluate several approximating functions without having to implement each one separately.

A caveat of this grid approach is that the maximum end point needs to be chosen in advance, effectively bounding the space of $a_0$ considered. 
While for many applications, interest usually lies in $a_0 \in [0, 1]$, even when one is interested in $a_0 > 1$, one usually has a good idea of the range of reasonable of values, since this information is also useful in specifying the prior $\pi_A(a_0\mid \delta)$.
In fact, prior information can be used to set the maximum grid value: let $p$ be a fixed probability and then set the maximum grid value $M$ such that
\[ \int_0^M \pi(a) da = p. \]
One can pick $p = 0.9999$, for instance, so as to have a high chance of not sampling any values of $a_0$ outside the grid.
This path is not explored here, however.

\subsubsection{Adaptively building the estimation grid}
\label{sec:adapt_grid}

Another approach is to build the estimation grid $\boldsymbol a^{\text{est}}$ adaptively.
Since $c(a_0)$ is convex, we need to make sure our grid covers the region where its derivative changes signs (if it does) when designing both $\boldsymbol a^{\text{est}}$ and $\boldsymbol a^{\text{pred}}$.
As discussed above, $l^\prime(a_0)$ changes signs at the same point as $c^\prime(a_0)$ does, and we shall exploit this to design our grids.
One can get an estimate of $c^\prime(a_0)$ directly from MCMC (see Equation~(\ref{eq:derivative_ca0}) in Appendix~\ref{sec:further_proofs}), since this is just the expected value of the log-likelihood under the power prior.
In practice this means that evaluating $l^\prime(a_0)$ comes essentially ``for free'' once one does the computations necessary to estimate $l(a_0)$.

In order to adaptively build the estimation grid, $\boldsymbol a^{\text{est}}$, we propose doing a bisection-type search.
First, let $m$ and $M$ be the grid the endpoints and $J$ be the budget on the total number of evaluations of $l(a_0)$.
Further, fix two real constants $v_1, v_2 > 0$.
In our computations we have used $v_1 = v_2 =  10$. 

\begin{enumerate}
 \item Initialize the variables $Z = \{ 0\}$, which will store the visited values of $a_0$,  $F = \{ 0\}$ which will store the values of $l(a_0)$ and $F^{\prime} = \{ \emptyset \}$ which will store the values of $l^{\prime}(a_0)$;
 \item Compute $l(m)$, $l^\prime(m)$, $l(M)$, and $l^\prime(M)$ and store these values in their respective variables.
\textbf{If} $\text{sgn}(l^\prime(m)) = \text{sgn}(l^\prime(M))$, construct $Z$ to be a regular grid of $J - 2$ values between $m$ and $M$ and estimate $l(\cdot)$ at those values, building $F$ and $F^{\prime}$ accordingly. 
\textbf{Else}, with $\text{sgn}(c^\prime(m)) \neq \text{sgn}(c^\prime(M))$, set $L^{(1)} = m$ and $U^{(1)} = M$ and  make $J = J - 1$.
Then, for the $k$-th iteration ($k > 1$):
 \begin{enumerate}
 \item Make $z^{(k)} = (L^{(k)} + U^{(k)})/2$, compute $l(z^{(k)})$, $l^{\prime}(z^{(k)})$ and store $Z \leftarrow z^{(k)}$, $F \leftarrow l(z^{(k)})$ and $F^{\prime} \leftarrow l^{\prime}(z^{(k)})$.
 \item Compare derivative signs: \textbf{if} $ \text{sgn}(l^\prime(z^{(k)})) = \text{sgn}(l^\prime(m))$ set $L^{(k + 1)} = z^{(k)}$ and $U^{(k + 1)} = U^{(k)}$. 
 Otherwise, set $L^{(k + 1)} = L^{(k)}$ and $U^{(k + 1)} = z^{(k)}$. 
 Compute $\delta^{(k)} = |z^{(k)} - z^{(k-1)}|$ and set $J = J - 1$.
 \textbf{If} $J = 0$, stop.
 \item \textbf{If} $J > 0$ but $\delta^{(k)}  <  v_1 m $, stop.
 \begin{enumerate}
  \item Compute $A^{(k)} = \max(0, z^{(k)} - v_2m)$ and $ B^{(k)} = \min(z^{(k)} + v_2m, M)$.
  \item Considering only the elements $z_i$ of $Z$ such that $A^{(k)} \leq z_i \leq B^{(k)}$, find the pair $(z_i, z_{i + 1})$ such that $|z_{i} -z_{i + 1}|$ is largest, make $z^{(k)}  = |z_{i} -z_{i + 1}|/2$, compute $l(z^{(k)})$, $l^{\prime}(z^{(k)})$ and store $Z \leftarrow z^{(k)}$, $F \leftarrow l(z^{(k)})$ and $F^{\prime} \leftarrow l^{\prime}(z^{(k)})$.
  Set $J = J -1$ and, if $J = 0$, stop.
 \end{enumerate}
\end{enumerate}
\end{enumerate}
 
Informally, the algorithm starts by approaching the point $\hat{a}$ at which $l^\prime(\hat{a}) = 0$ and storing the values encountered on that path.
Because we do not want to waste all of the computational budget in evaluating too small a neighbourhood around $\hat{a}$, we use $v_1$ to control the size of this neighbourhood.
Then, if the computational budget of $J$ evaluations has not yet been exhausted, we ``plug the gaps'' in our collection of values of $a_0$, $Z$.
Because these gaps matter more closer to the region around $\hat{a}$, we use $v_2$ to control the size of the neighbourhood where we plug the gaps.
 
The algorithm discussed in this section shares many similarities with power (tempered) posterior methods for computing marginal likelihoods~\citep{Gelman1998, Friel2008}.
A key difference is that while the aforementioned methods are concerned with estimating the normalising constant for a single value of $a_0$ ($a_0 = 1$), here we are interested in approximating the whole $c(a_0)$ curve.

\subsubsection{When estimating marginal likelihoods directly is impractical}
\label{sec:adapt_grid_derivOnly}

For very complex, parameter-rich models, it might be the case that estimating $l(a_0)$ -- or $c(a_0)$ -- at a few values of $a_0$ may still be very computationally costly.
It may also be the case that the posterior density $L(D_0 \mid \theta)\pi(\theta)$ is costly to compute.
An alternative approach is to only evaluate $l^\prime(a_0)$ instead, which should be cheaper, and then obtain an approximation of $l(a_0)$~\textit{via} quadrature.
To this end, the algorithm presented in Section~\ref{sec:adapt_grid} can be used with very little change, namely by simply not evaluating (estimating) $l(\cdot)$ at points $Z$.
One can then approximate $l^\prime(a_0)$ -- instead of $l(a_0)$ -- in the same way by estimating an approximate curve $g_\xi$, evaluating predictions on a fine grid (say, $K = 20, 000$) and then using midpoint integration to obtain an approximation of $l(a_0)$.
We analysed this approach on a limited set of examples and found that it yielded less accurate approximations when compared to the method approximating $c(a_0)$ directly (see Appendix~\ref{sec:derivative_only}) and thus did not pursue the matter further.

\subsection{Computational details}
\label{sec:comp_details}

Easily extendable computer code for reproducing all of the steps as well as implementing new models is available from~\url{https://github.com/maxbiostat/propriety_power_priors}.

\subsubsection{Markov Chain Monte Carlo}
\label{sec:mcmc}

The vast majority of the models discussed in this paper lead to posterior distributions which cannot be written in closed-form and hence we must resort to numerical methods.
Here we employ Hamiltonian -- or Hybrid -- Monte Carlo (HMC), implemented in the Stan programming language~ \citep{Carpenter2017} in order to estimate the expectations of interest.
An excellent review of HMC can be found in~\cite{Neal2011}.
Unless stated otherwise, all of the analyses reported here are the result of running four independent chains of 2000 iterations each, with the first 1000 removed as burn-in/warmup.
Convergence was checked for by making sure all runs achieved a potential scale reduction factor (PSRF, $\hat{R}$) smaller than $1.01$.
For all expectations, we ensured Monte Carlo error (MCSE) was smaller than 5\% of the posterior standard deviation.

\subsubsection{Bridge sampling}
\label{sec:bridge}

Our approach relies heavily on estimates of $l(a_0)= \log(c(a_0))$, which are (log) marginal likelihoods, at selected values of $a_0$.
We employed bridge sampling~\citep{Meng1996,Meng2002} to compute marginal likelihoods using the methods implemented in the R package~\textbf{bridgesampling}~\citep{Gronau2017}.

Let $\boldsymbol{\Theta} \subseteq \mathbb{R}^d$ and let $(\boldsymbol{\Theta}, \mathcal{F}, P)$ be a probability space.
Suppose $P$ admits a density $p$ and consider computing
\begin{equation}
\label{eq:norm_const}
\nonumber
 Z = \int_{\boldsymbol{\Theta}} q(t) dP(t),
\end{equation}
where $p(\theta) = q(\theta)/Z$.
The quantity $Z$ is usually called the normalising constant of $p$, and finds use in many applications in Statistics, particularly in Bayesian Statistics.
In a Bayesian context, it is usual to compute the marginal likelihood $m(\boldsymbol{X} \mid \mathcal{M}) := \int_{\boldsymbol{\Theta}} L(\boldsymbol{X} \mid t, \mathcal{M} ) d\pi(t)$, where $\pi$ is the prior measure, and this quantity is the~\textbf{evidence} in favour of model $\mathcal{M}$. 
Computing $Z$ for most models of interest involves computing a high-dimensional integral which can seldom be solved in closed-form and is thus a difficult numerical task that requires specialised techniques.
As before, denote $f_{a_0}(D_0; \theta) = L(D_0 \mid \theta)^{a_0} \pi(\theta)$.
Here we are interested in estimating
\begin{equation}
 \label{eq:marginal_like_power_prior}
 \nonumber
 c(a_0, D_0) = \int_{\boldsymbol{\Theta}} f_{a_0}(D_0; \theta) d\theta.
\end{equation}
While previously we have omitted the dependence of $c(a_0)$ on the data $D_0$ for clarity, here we shall write the full expression for completeness.

The method proposed initially by~\cite{Meng1996} and extended by~\cite{Meng2002} gives the estimator
\begin{equation}
\label{eq:brige_estimator}
  \hat{c}(a_0, D_0) = \frac{N^{-1}\sum_{j=1}^N h(\tilde{\theta}_j) f_{a_0}(D_0; \tilde{\theta}_j) }{M^{-1}\sum_{i=1}^M h(\theta_i^\ast) g(\theta_i^\ast) },
\end{equation}
where $h(\cdot)$ is the bridge distribution and $g(\cdot)$ is a proposal density.
We then let $\tilde{\boldsymbol{\theta}} = \{\tilde{\theta}_1, \tilde{\theta}_2, \ldots, \tilde{\theta}_N \}$ and $ \boldsymbol{\theta}^\ast = \{ \theta_1^\ast, \theta_2^\ast, \ldots, \theta_M^\ast  \}$ are sets of $N$ and $M$ samples from $f_{a_0}(D_0; \theta)$ and $h$, respectively.

The performance of the estimator depends on the optimal choice of $h$, which in turn depends on $c(a_0, D_0)$, the target quantity.
To overcome this difficulty we use an iterative procedure to obtain the estimate from an initial guess $\hat{c}(a_0, D_0)^{(0)}$:
\begin{equation}
 \label{eq:iterative_estimate_marglike}
  \hat{c}(a_0, D_0)^{(t + 1)} = \frac{N}{M}\frac{\sum_{j=1}^N \frac{w(\tilde{\theta}_j)}{aw(\tilde{\theta}_j) + (1-a)\hat{c}(a_0, D_0)^{(t)} } }{\sum_{i=1}^M \frac{1}{aw(\theta_i^\ast) + (1-a)\hat{c}(a_0, D_0)^{(t)} }},
\end{equation}
where $w(\theta) = f_{a_0}(D_0; \theta)/g(\theta)$ and $a =  M/(M + N)$.
Numerically stable routines for computing~(\ref{eq:iterative_estimate_marglike}) are implemented in the package~\textbf{bridgesampling}~\citep{Gronau2017}.
As a note, the estimator in~(\ref{eq:brige_estimator}) assumes that the samples are independent and identically distributed, which is not the case when samples are obtained~\textit{via} MCMC, and hence $M$ and $N$ are replaced with estimates of the effective sample size (ESS).
Since Stan achieves high efficiency ($\text{ESS}/\text{\# samples}$) for most models considered here, this poses no problem.
Here we use the default settings of the algorithm available in~\textbf{bridgesampling}, meaning we take $g(\cdot)$ to be a multivariate proposal distribution.
As explained by~\cite{Gronau2017} the bridge sampling algorithm is robust to the tail behaviour of the target and proposal distributions as long as $h$ is optimal, which is the case here.

\subsubsection{Generalised additive models}
\label{sec:gam}

The approach in Section~\ref{sec:efficient_computation_ca0} (see also Section~\ref{sec:bridge} above) allows us to obtain a set of $J$ pairs $(a_{0i}, \hat{l}_i)$ from which the approximating function $g_{\boldsymbol{\xi}}(a_0) \approx l(a_0)$ can be estimated. 
We now detail our tool of choice to construct $g_{\boldsymbol{\xi}}$, the generalised additive model (GAM)~\citep{Wood2017}.

A GAM is a model for the conditional expectation of the dependent variable, $\mu_i := E[Y_i]$, of the form
\begin{equation}
 g(\mu_i) = \boldsymbol{X}_i^\ast\boldsymbol{\eta} + \sum_{k = 1}^q f_i(\boldsymbol{X}_i^\ast)
\end{equation}
where $g$ is a link function, $\boldsymbol{\eta}$ is a vector of coefficients for the parametric components of the model and the $f_i$ are smooth functions of the covariates.
In particular, here we are interested in the model
\begin{equation}
\label{eq:gam_one_smooth}
  \hat{l}_i = \Delta + \sum_{k = 1}^q b_k(a_{0i})\beta_k  + \epsilon_i,
\end{equation}
and where $b_k$ is the $k$-th basis function, $\Delta$ is an intercept and we assume $\epsilon_i\sim \operatorname{Normal}(0, \tau)$.
We employ the routines in the \textbf{mgcv} package~\citep{Wood2011} in R to fit GAMs and make predictions.

In our applications we employed $q = J$.
This choice leads to overfitting, which in other settings would be undesirable.
In our situation, however, overfitting is not a problem because the end goal is to predict the value of $l(a_0)$ within the measured range of the covariate $a_0$, $[m, M]$.
Once we have fitted the model in~(\ref{eq:gam_one_smooth}), we have our approximating function $g_{\boldsymbol{\hat{\xi}}}$, where $\boldsymbol{\hat{\xi}} = \{ \hat{\Delta}, \hat{\beta_1}, \hat{\beta_2}, \ldots, \hat{\beta_q}, \hat{\tau}\}$, which we can in turn use to predict $l(a_0)$ over a grid of $K$ points covering  $[m, M]$.

\section{Illustrations}
\label{sec:illustrations}

In this section we discuss applications of the normalised power prior.
We first discuss four examples where $c(a_0)$ is known in closed-form and use these as a benchmark for the approximations discussed in this paper.
Then we move on to explore two regression examples where the normalising constant is not known in closed-form and thus only the approximation is available.  

In all examples we employed a Beta prior on $a_0$ with parameters $\eta = \nu = 1$ and thus restricted attention to $a_0 \in [0, 1]$ when approximating the normalising constant -- i.e., we used $M = 1$.
The exception was the Gaussian example in Section~\ref{sec:gaussian_illus} where we used $M=10$ in order to study the method in the non-monotonic case.
For all exampĺes we used $m = 0.05$ and employed budget of $J = 20$ evaluations of $l(a_0)$ via bridge sampling.   

\subsection{Bernoulli likelihood}
\label{sec:reproduce_N2009}

In this section we revisit the Bernoulli example of~\cite{Neuenschwander2009} and show how the approximation scheme proposed in Section~\ref{sec:efficient_computation_ca0} can be used, taking advantage of the fact that $c(a_0)$ is known exactly for this example.
The historical data consist of $N_0$ Bernoulli trials $x_{0i} \in \{0,1\}$.
Suppose there were $y_0 = \sum_{i=1}^{N_0}x_{0i}$ successes.
The model reads
\begin{align*}
 \theta &\sim \operatorname{Beta}(c, d), \\
 x_{0i} \mid \theta &\sim \operatorname{Bernoulli}(\theta).
\end{align*}
This leads to a Beta posterior distribution for $\theta$,
\begin{equation}
 \label{eq:bernoulli_posterior}
 p(\theta \mid N_0, y_0, a_0) \propto \theta ^{a_0 y_0 + c - 1} (1-\theta)^{a_0 (N_0 -y_0) + d - 1},
\end{equation}
and hence~\citep{Neuenschwander2009}:
\begin{equation}
 \label{eq:cA0_bernoulli}
 c(a_0) = \frac{\mathcal{B}(a_0 y_0 + c, a_0 (N_0 -y_0) + d)}{\mathcal{B}(c, d)},
\end{equation}
where $\mathcal{B}(w, z) = \frac{\Gamma(w)\Gamma(z)}{\Gamma(w + z)}$.
The derivative, $c^\prime(a_0)$, evaluates to
\begin{equation}
  \label{eq:cA0_prime_bernoulli}
c^\prime(a_0) = \frac{\mathcal{B}(z_0, w_0) \left(y_0 \left[\psi_0(w_0) - \psi_0(z_0) \right] + N_0 \left[ \psi_0(z_0) - \psi_0(w_0 + z_0) \right] \right) }{\mathcal{B}(c, d)},
\end{equation}
where $z_0 = a_0 y_0 + c $ and $w_0 = a_0 (N_0 -y_0) + d$.
If one observes new data $D = (N, y) $, one can then compute a posterior $p(\theta \mid a_0, D_0, D)$.
In the situation where one lets $a_0$ vary by assigning it a prior $\pi_A ( \cdot \mid \delta)$, one can write the marginal posterior for $a_0$ explicitly~\cite[Eq. 8]{Neuenschwander2009}:
\begin{equation}
 \label{eq:marginal_posterior_a0_Bernoulli}
 p(a_0 \mid D_0, D) \propto c(a_0) \pi_A(a_0 \mid \delta) \mathcal{B}(a_0 y_0 + y  + c -1, a_0(N_0-y_0) + (N -y) + d -1 ).
\end{equation}

\cite{Neuenschwander2009} thus consider the problem of estimating the probability of response $\theta$ in a survey where $y$ of the $N$ individuals are responders and $N -y$ are non-responders.
They consider four scenarios, that vary the historical ($D_0 = \{N_0, y_0\}$) and current data ($D = \{N, y\}$), detailed in Table~\ref{tab:results_Bernoulli}.
They employ flat Beta priors $\pi_A(a_0\mid \delta)$ with parameters $\eta = \nu = 1$ and $\pi(\theta)$ with parameters $c = d = 1$, which we also adopt here.

First, we show a sensitivity analysis where we computed the prior and posterior distributions for the quantity of interest $\theta$ for various ($J=20$) values of $a_0$ in order to gauge how the discount factor affects the inferences reached. 
In Figure~\ref{fig:Bernoulli_sensanalysis} we show the distribution of $\theta$ for various values of $a_0$ in each scenario.
When the historical and current data are compatible ($y_0/N_0 = y/N$) as in scenarios 1 and 3, we see that the prior uncertainty encompasses the posterior for all values of $a_0$.
In contrast, when there is incompatibility between the historical and current data sets, we see that the prior and posterior intervals stop overlapping for moderate values of $a_0$, an effect more prominent the larger $N$ is.
For scenario 2  we see overlap up until $a_0 \approx 0.30$, while for scenario 4, with more data, incompatibility starts to arise much earlier, around $a_0 \approx 0.05$.

\begin{figure}[!ht]
\begin{center}
\includegraphics[scale=0.6]{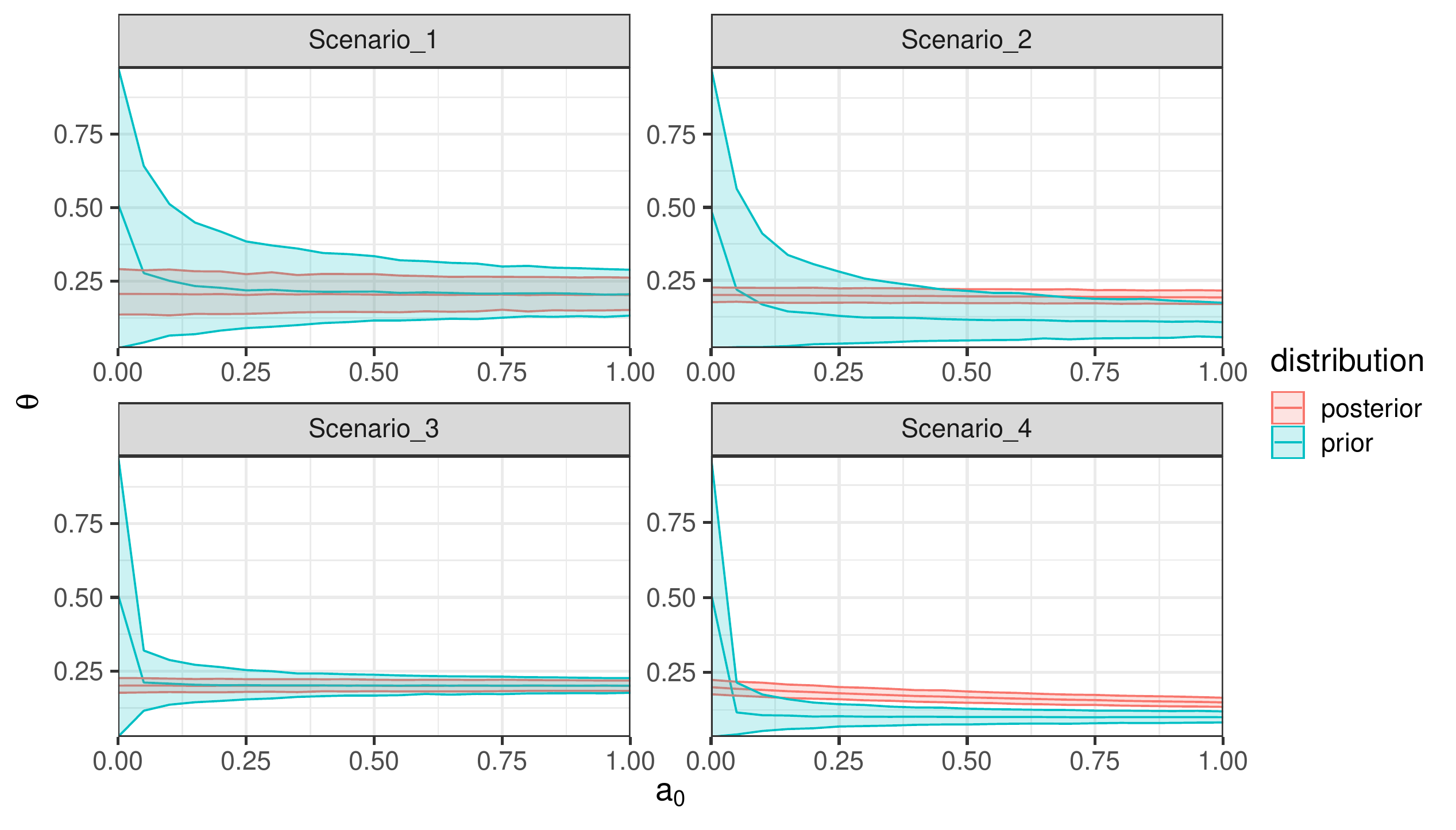}
\end{center}
\caption{\textbf{Sensitivity analysis for the Bernoulli example}.
We show the prior and posterior distribution for $\theta$ as the discounting factor $a_0$ varies.
Colours show the distribution in question:  the prior $\pi_{a_0}(\theta) = L(D_0 \mid \theta)^{a_0}\pi(\theta)$ or the posterior $p_{a_0}(\theta) = L(D\mid \theta)\pi_{a_0}(\theta)$.
}
\label{fig:Bernoulli_sensanalysis}
\end{figure}

We show that for all scenarios considered, posterior estimates of the response proportion $\theta$ and the power prior scalar $a_0$ are extremely consistent between the approximate normalisation and the exact normalisation given in~(\ref{eq:cA0_bernoulli}).
For all the scenarios considered, $l(a_0)$ looks approximately linear in $a_0$ as shown in Figure~\ref{sfig:ca0_Bernoulli}.
While we used a relatively fine grid ($K = 20, 000$) to create the $l(a_0)$ dictionary, we found that smaller values also gave good performance (data not shown, see below).

\begin{table}[!ht]
\caption{\textbf{Bernoulli example}.
We compare estimates of both the response proportion $\theta$ and the power prior scalar $a_0$ using the unnormalised power prior, the exactly normalised prior (Eq.~\ref{eq:cA0_bernoulli}) and an approximation obtained according to the method in Section~\ref{sec:efficient_computation_ca0}.
We approximated $l(a_0)$ using $J = 20$ grid points for estimation and $K = 20, 000$ points for prediction.}
\begin{center}
{ \small
 \begin{tabular}{cccccc}
\hline
            Scenario                &     Data            &   Parameter    & Unnormalised      & Normalised        & App. normalised \\
\hline                            
\multirow{2}{*}{Scenario 1} & \multirow{2}{*}{$\frac{y_0}{N_0} = \frac{20}{100}$, $\frac{y}{N} = \frac{20}{100}$} & $\theta$ & 0.21 (0.13, 0.29) & 0.20 (0.15, 0.27) & 0.20 (0.14, 0.27)          \\
                            &  & $a_0$  & 0.02 (0.00, 0.07) & 0.57 (0.07, 0.98) & 0.58 (0.07, 0.98)          \\
\multirow{2}{*}{Scenario 2} &  \multirow{2}{*}{$\frac{y_0}{N_0} = \frac{10}{100}$, $\frac{y}{N} = \frac{200}{1000}$} & $\theta$ & 0.20 (0.18, 0.23) & 0.20 (0.17, 0.22) & 0.20 (0.17, 0.22)          \\
                            &  & $a_0$  & 0.03 (0.00, 0.10) & 0.36 (0.02, 0.93) & 0.37 (0.03, 0.92)          \\
\multirow{2}{*}{Scenario 3} &  \multirow{2}{*}{$\frac{y_0}{N_0} = \frac{200}{1000}$, $\frac{y}{N} = \frac{200}{1000}$}& $\theta$ & 0.20 (0.18, 0.23) & 0.20 (0.18, 0.22) & 0.20 (0.18, 0.22)          \\
                            &  & $a_0$  & 0.00 (0.00, 0.01) & 0.57 (0.06, 0.98) & 0.59 (0.09, 0.98)          \\
\multirow{2}{*}{Scenario 4} &  \multirow{2}{*}{$\frac{y_0}{N_0} = \frac{100}{1000}$, $\frac{y}{N} = \frac{200}{1000}$} & $\theta$ & 0.20 (0.18, 0.23) & 0.20 (0.17, 0.22) & 0.20 (0.18, 0.23)          \\
                            &  & $a_0$  & 0.00 (0.00, 0.01) & 0.05 (0.00, 0.15) & 0.04 (0.00, 0.16)\\
\hline
\end{tabular}
}
\end{center}
\label{tab:results_Bernoulli}
\end{table}

The marginal posteriors of $a_0$ obtained for each scenario are shown in Figure~\ref{fig:marginal_a0_Bernoulli}.
As shown in Table~\ref{tab:results_Bernoulli}, the approximately normalised power prior are in close agreement with the closed-form solution, for a variety of shapes the distribution takes across scenarios.
In particular, scenarios 1 and 3 are designed such that posterior estimates of $a_0$ should be around $1/2$ in order to reflect the fact that current data is compatible with historical data.
On the other hand, scenarios 2 and 4 are designed such that there is mild incompatibility between historical and current data, and this is reflected in the properly normalised posteriors for $a_0$, whereas the unnormalised posteriors yield counter-intuitive results.

Looking closely at Figure~\ref{fig:marginal_a0_Bernoulli}d, however, we notice that while the mean and BCI of the approximately normalised posterior are not significantly different from the exactly normalised distribution, the shape of the marginal posterior density for $a_0$ does show some inconsistencies.
This serves as a warning that the approximate normalisation does not work equally well in all situations, and may be susceptible to non-linearities in the sense that a small error in approximating $c(a_0)$ might have a big impact on the estimates.

\begin{figure}[!ht]
\begin{center}
 \hfill
\subfigure[Scenario 1]{\includegraphics[width=7cm]{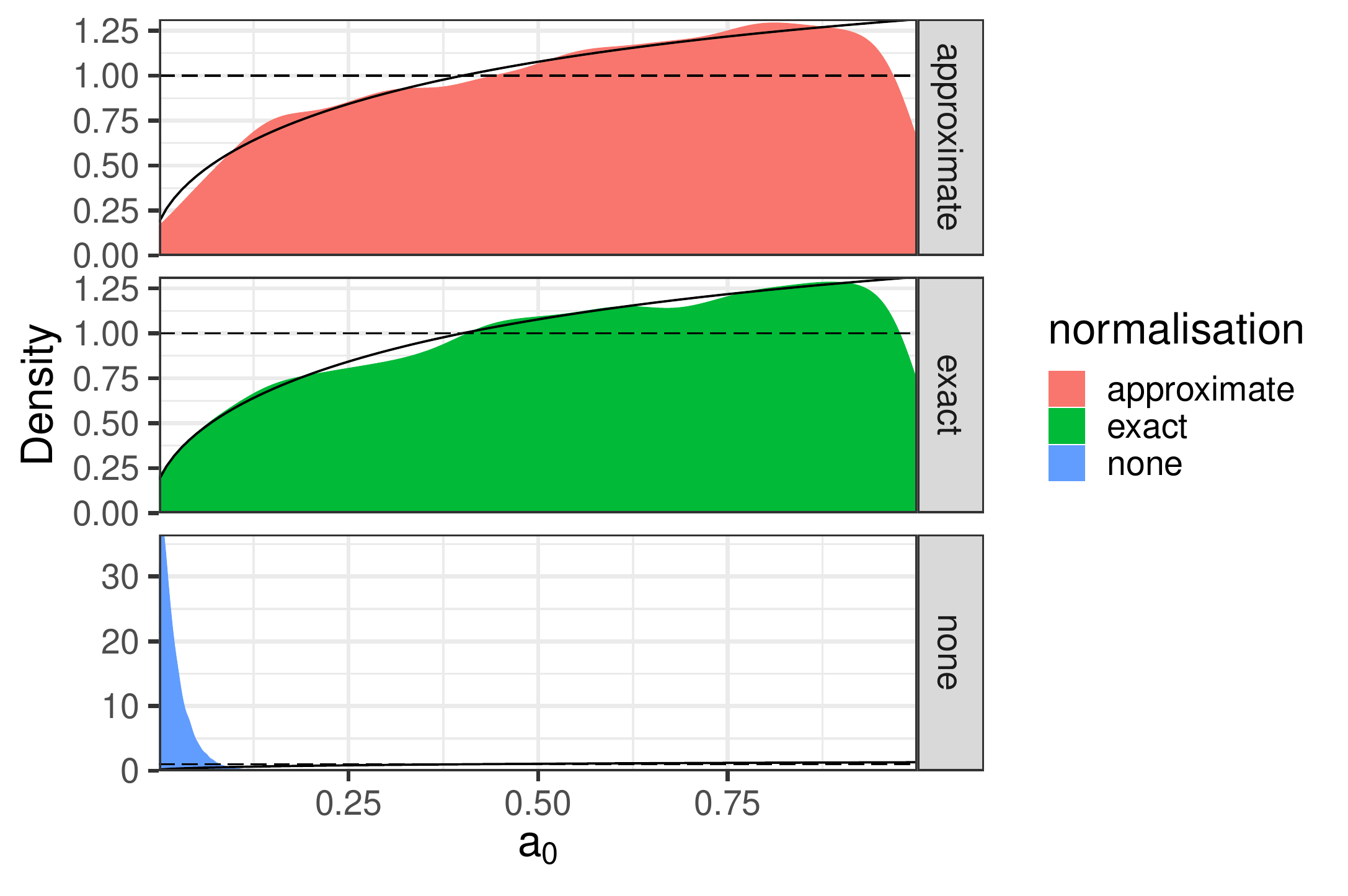}}
\hfill
\subfigure[Scenario 2]{\includegraphics[width=7cm]{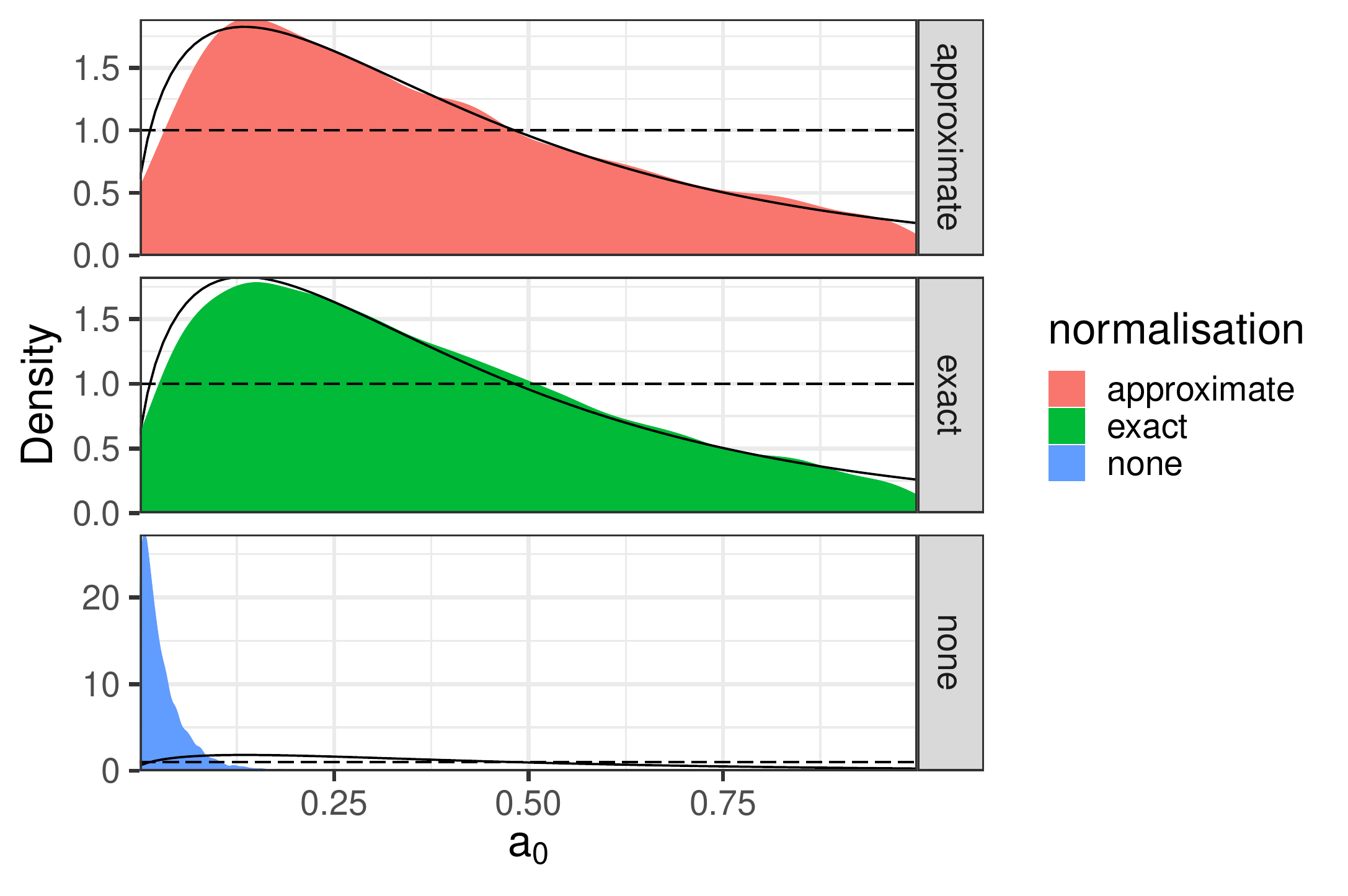}}\\
\hfill
\subfigure[Scenario 3]{\includegraphics[width=7cm]{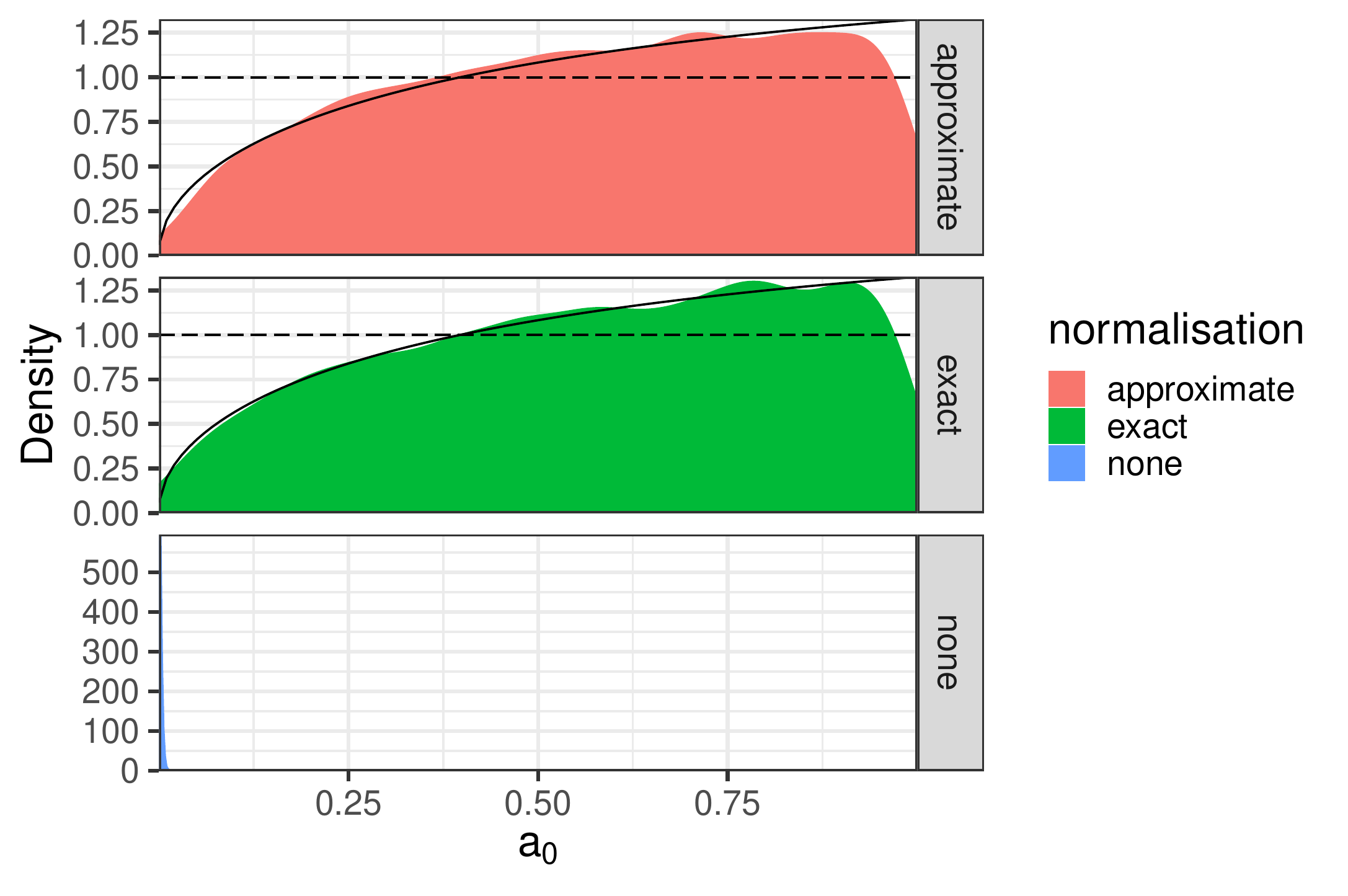}}
\hfill
\subfigure[Scenario 4]{\includegraphics[width=7cm]{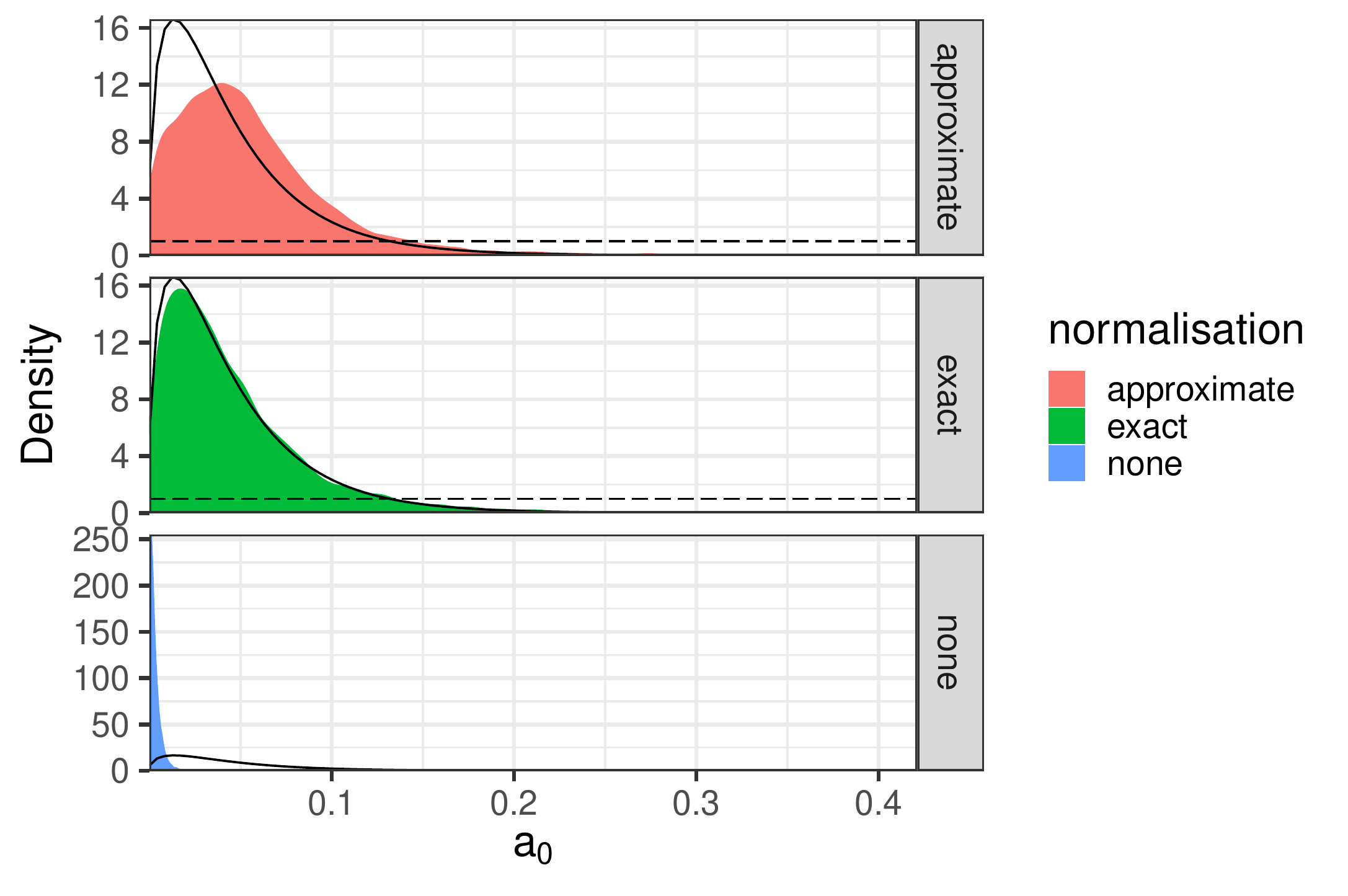}}
\hfill
\end{center}
\caption{\textbf{Marginal distributions of $a_0$ for the Bernoulli example}.
We show the marginal posterior of $a_0$ as given by~(\ref{eq:marginal_posterior_a0_Bernoulli}), normalised~\textit{via} quadrature.
Colours (and horizontal tiles) show the normalisation method used: none (unnormalised), exact (normalised) or approximate.
Horizontal dashed line marks the Beta prior with parameters $\eta = \nu = 1$.
Please note that the x-axes differ between panels.
}
\label{fig:marginal_a0_Bernoulli}
\end{figure}

\subsection{Poisson likelihood}
\label{sec:poisson_illustration}

Now consider another simple discrete example, the modelling of counts.
Suppose the historical data consist of $N_0$ observations $y_{0i} \in \{0, 1, \ldots \}$, assumed to come from a Poisson distribution.
For simplicity, we will again consider the conjugate case:
\begin{align*}
 \lambda &\sim \operatorname{Gamma}(\alpha_0, \beta_0),\\
 y_{0i} \mid \lambda &\sim \operatorname{Poisson}(\lambda).
\end{align*}
The posterior distribution is 
\begin{equation}
 p(\lambda \mid \boldsymbol y_0) \propto \frac{1}{\boldsymbol {p^\prime}^{a_0} } \lambda^{a_0\boldsymbol s} \exp(-a_0 N_0 \lambda) \times \lambda^{\alpha_0-1} \exp(-\beta_0\lambda),
\end{equation}
where $\boldsymbol s := \sum_{i=0}^{N_0} y_{0i}$ and $\boldsymbol p^\prime := \prod_{i = 0}^{N_0} y_{0i}!$, leading the closed-form expression
\begin{equation}
 \label{eq:cA0_poisson}
 c(a_0) = \frac{\beta_0^{\alpha_0}}{\Gamma(\alpha_0)}\frac{1}{\boldsymbol {p^\prime}^{a_0} } \frac{\Gamma(a_0\boldsymbol s + \alpha_0)}{\left( a_0N_0 + \beta_0 \right)^{a_0\boldsymbol s + \alpha_0} }.
\end{equation}

For this model we have (see Remark~\ref{rmk:discrete_decreasing}):
\begin{equation}
 \label{eq:cA0_prime_poisson}
 c^\prime(a_0) = \left[ -\log(\boldsymbol p^\prime) - \frac{N_0 (\alpha_0 + \boldsymbol s a_0) }{\beta_0 + N_0a_0} - \boldsymbol s \log(\beta_0 + N_0a_0) + \boldsymbol s \psi_0(\alpha_0 + \boldsymbol s a_0) \right] c(a_0).
\end{equation}

We can use this example to study the quality of the approximation to $c(a_0)$ as the number of grid points $K$ increases.
For the experiment in this section we simulated $N_0 = 200$ historical data points $\boldsymbol y_0$ with $\lambda = 2$ and $N = 100$ current data points $\boldsymbol y$ with the same rate parameter.
The prior hyperparameters are $\alpha_0 = \beta_0 = 2$.
Figure~\ref{fig:poisson} shows the resulting marginal posteriors for $a_0$ and $\lambda$ using several values of the grid size, $K$, in .
Even for relatively small values of $K$, such as $K = 50$, the approximately normalised posteriors are very similar to the posterior obtained with exact normalisation, both for $a_0$ and $\lambda$.

\begin{figure}[!ht]
\hfill
\subfigure[$a_0$]{\includegraphics[width=7cm]{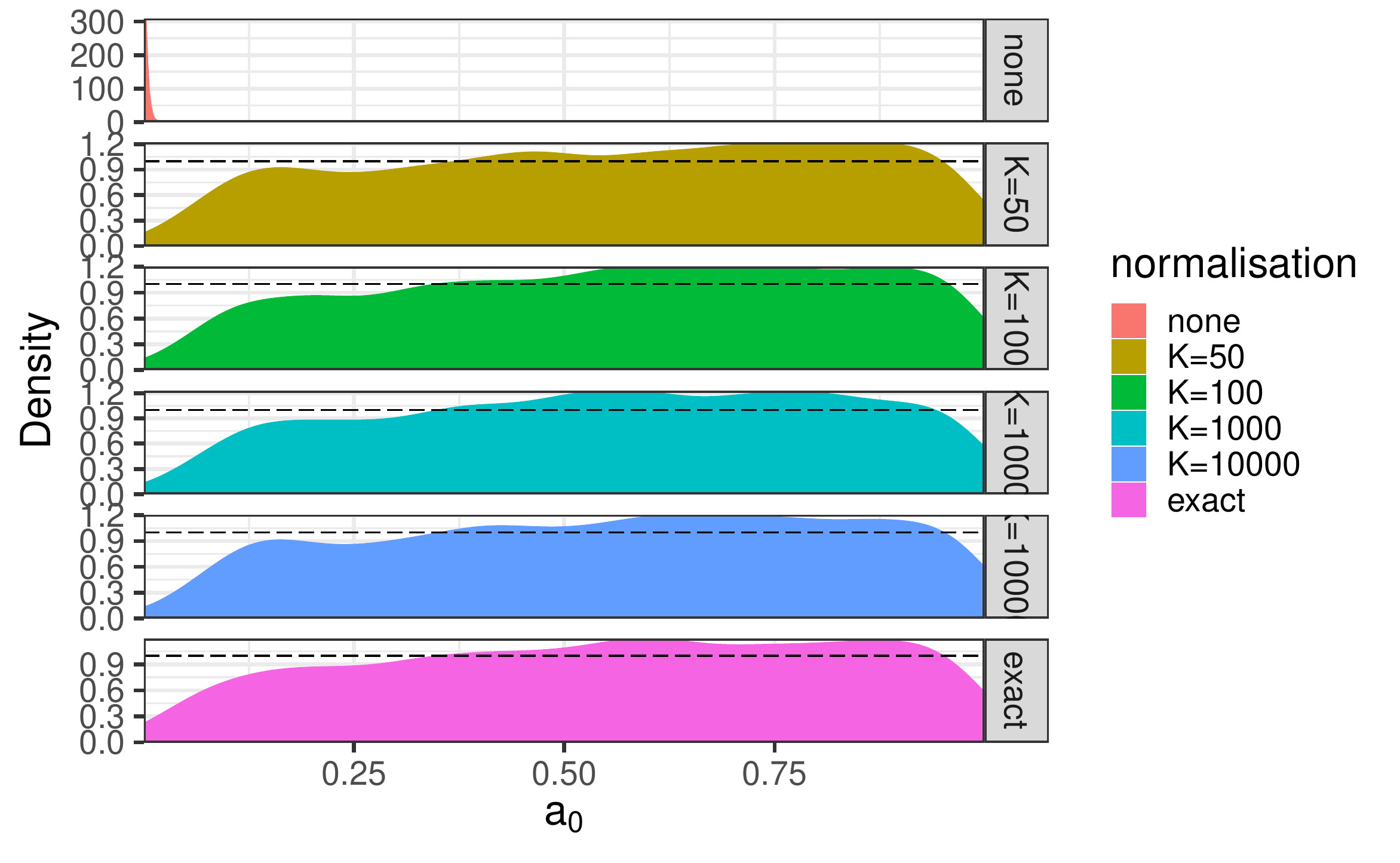}}
\hfill
\subfigure[$\lambda$]{\includegraphics[width=7cm]{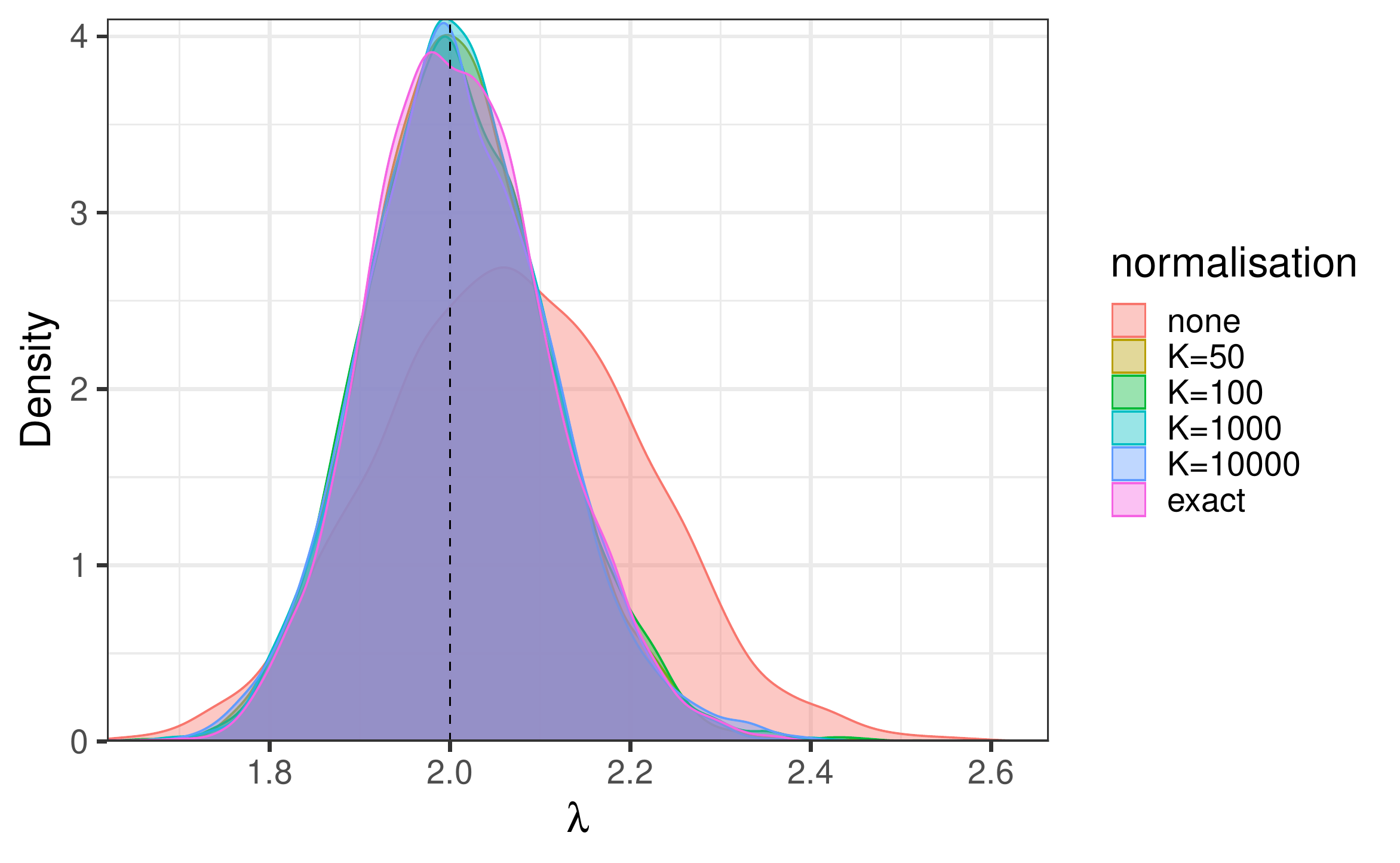}}
\hfill
\subfigure[$\lambda$]{\includegraphics[width=7cm]{parameter_posterior_Poisson.pdf}}
\hfill
\caption{\textbf{Results for the Poisson example}.
Panels (and colours) correspond to various values of the grid size, $K$, as well as the results with no normalisation.
In panel (a) we show the marginal posterior for $a_0$, using horizontal dashed lines to show the prior density of a $\operatorname{Beta}(\eta = 1, \nu = 1)$.
}
\label{fig:poisson}
\end{figure}

\subsection{Gaussian likelihood with unknown mean and variance}
\label{sec:gaussian_illus}

Now we move on to study a case where $c(a_0)$ is non-monotonic and thus presents a more challenging setting.
Suppose one has $N_0$ historical observations $y_{i0} \in \mathbb{R}, i = 1, \ldots, N_0$, which come from a Gaussian distribution with parameters $\mu$ and $\tau$.
Here we will choose a normal-Gamma conjugate model:
\begin{align*}
 \tau &\sim \operatorname{Gamma}(\alpha_0, \beta_0),\\
 \mu &\sim \operatorname{Normal}(\mu_0, \kappa_0\tau ),\\
 y_{0i} \mid \mu, \tau &\sim \operatorname{Normal}(\mu, \tau),
\end{align*}
where the normal distribution is parametrised in terms of mean and precision (see below for a different parametrisation).
The posterior distribution is again a normal-Gamma distribution and the normalising constant is
\begin{align}
 \label{eq:cA0_gaussian}
 c(a_0) &= \frac{\Gamma(\alpha_n)}{\Gamma(\alpha_0)}\frac{\beta_0^{\alpha_0}}{\beta_n^{\alpha_n}} \left(\frac{\kappa_0}{\kappa_n} \right)^2 (2\pi)^{-N_0 a_0/2},\\
 \nonumber
 \alpha_n &= \alpha_0 + \frac{1}{2}a_0N_0, \\
 \nonumber
 \kappa_n &= \kappa_0 + a_0N_0, \\
 \nonumber
 \beta_n  &= \beta_0 + \frac{1}{2}\left( a_0\sum_{i=1}^{N_0}(y_{0i}-\bar{y})^2 + \left(\kappa_0 a_0 N_0 (\bar{y}-\mu_0)^2\right)/\kappa_n \right),
\end{align}
with $\bar{y} = N_0^{-1}\sum_{i=1}^{N_0} y_{0i}$.
In Appendix~\ref{sec:ca0_norm_deriv}, we give a closed-form expression for $c^\prime(a_0)$ and characterise the point of inflection of $c(a_0)$ by giving the conditions for $c^\prime(a_0) = 0$.

To make the discussion concrete, we generate $N_0 = 50$ data points from a Gaussian distribution with parameters $\mu = -0.1$ and  $\tau = 10^{6}$.
We construct the Gamma prior on $\tau$ with $\alpha_0 = \beta_0 = 1$ and assign a Gaussian prior on $\mu$, with parameters $\mu_0 = 0$ and $\kappa_0 = 5$.
This
choice of hyperparameters leads to a function $c(a_0)$ -- Equation~(\ref{eq:cA0_gaussian}) -- that resembles a concave up parabola (Figure~\ref{fig:gaussian_results}a). 
We then generate $N = 200$ new points from the same distribution to be used as current data.
The points show the values of $l(a_0)$ and $l^\prime(a_0)$ estimated using the algorithm described in Section~\ref{sec:adapt_grid}, which exploits the derivatives of $c(a_0)$ to place more points closer to the region where $c^\prime(a_0)$ (and $l^\prime(a_0)$) changes signs.

We show the resulting marginal posteriors for $\mu$ and $\tau$ as well as $a_0$ under no normalisation, exact and approximate normalisation with various $K$ in Figure~\ref{fig:gaussian_results}b and~\ref{fig:gaussian_results}c.
The first observation is that approximations with $K > 100$ seem to produce marginal posteriors for $a_0$ that resembles the exactly normalised distribution quite closely, even in this setting, where $c(a_0)$ is non-linear.

In terms of parameter posteriors, we find that the posterior is not very sensitive to the value of $a_0$, as shown by the overlap between marginal posteriors with no normalisation as well as exact and approximate normalisation.
Even in this setting the approximately normalised marginal posteriors match their exact counterparts closely.

\begin{figure}[!ht]
\hfill
\subfigure[(log) normalising constant ]{\includegraphics[width=7cm]{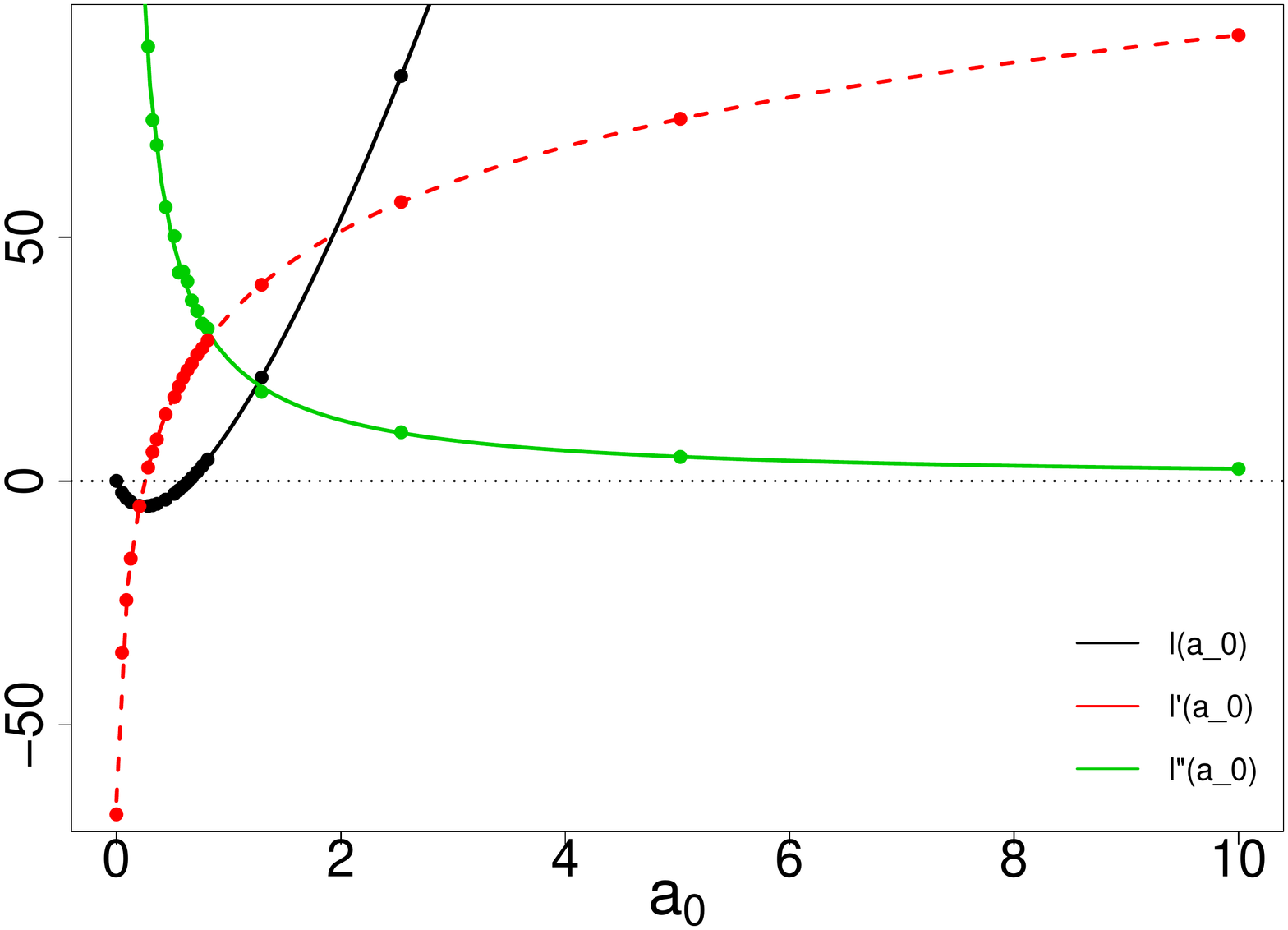}}
\hfill
\subfigure[$a_0$ posterior]{\includegraphics[width=7cm]{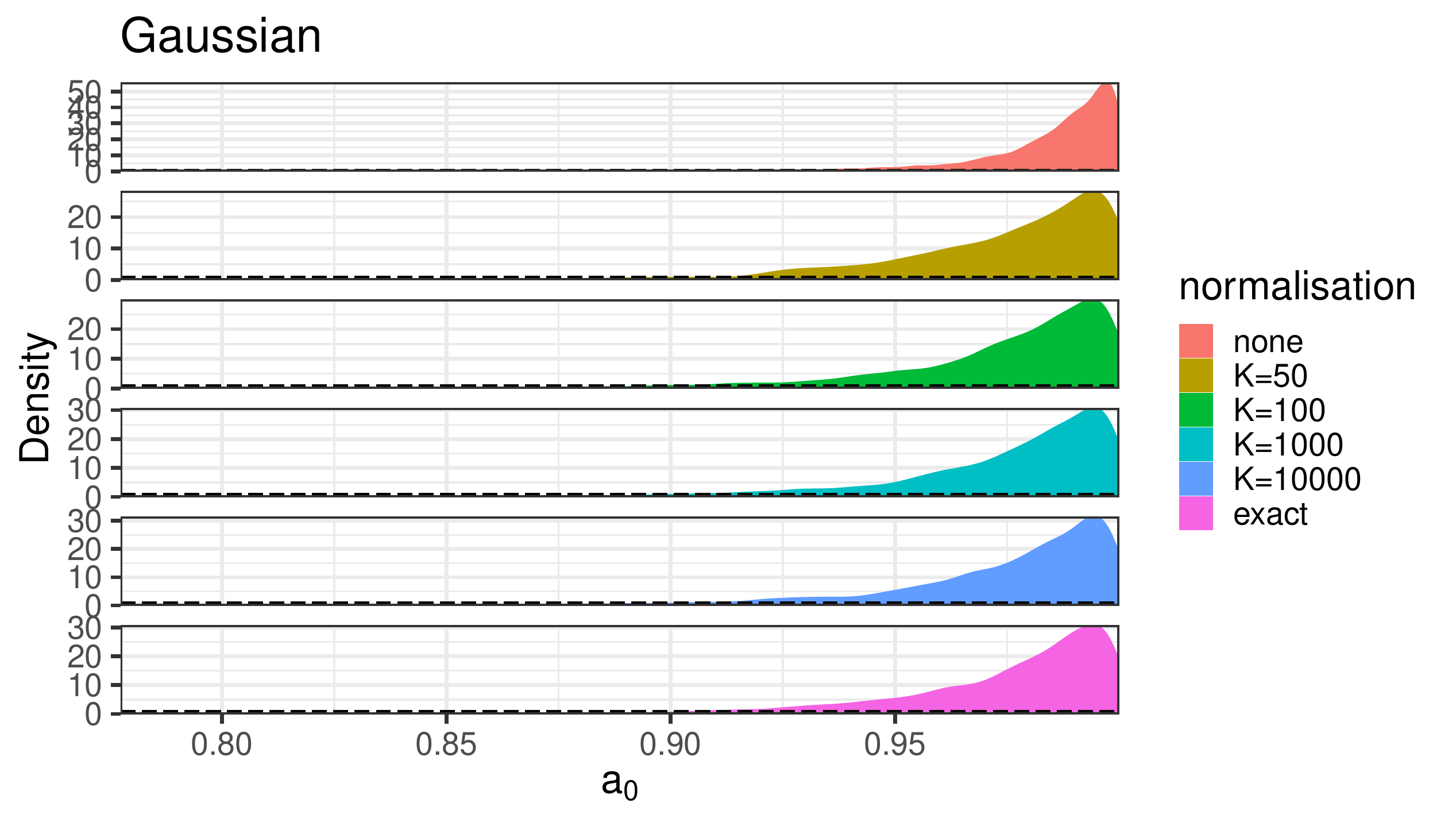}}\\
\hfill
\subfigure[Parameter posteriors]{\includegraphics[width=7cm]{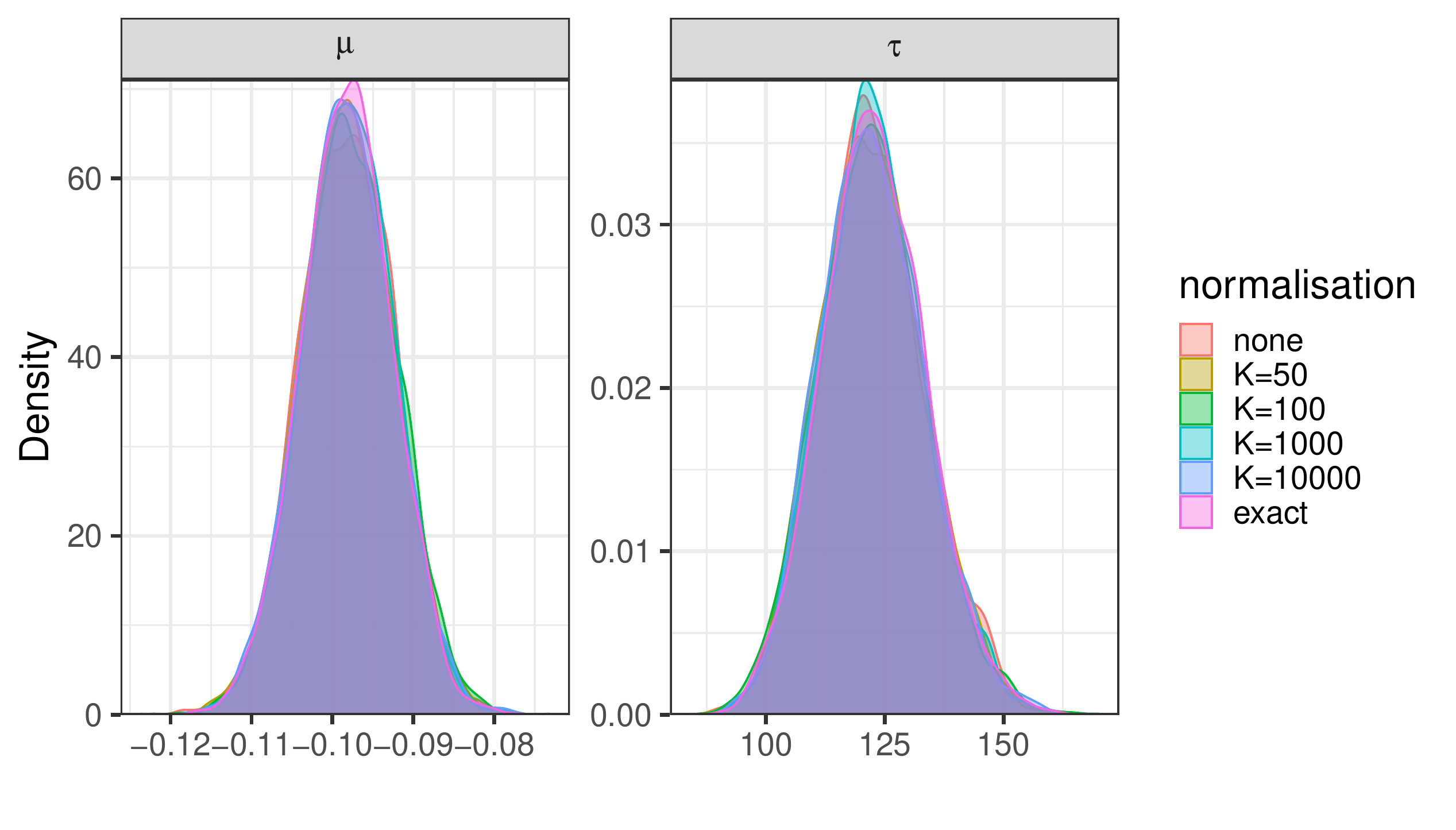}}
\hfill
\caption{\textbf{Results for the Gaussian example}.
In panel (a) we show $l(a_0) := \log(c(a_0))$ (black) and its first two derivatives (red and green, respectively).
Points show the $J = 20$ estimates of $l(a_0)$, $l^\prime(a_0)$ and $l^{\prime\prime}(a_0)$ obtained using the algorithm in Section~\ref{sec:adapt_grid}.
In panel (b)  we show the marginal posterior of $a_0$ under no normalisation, exact normalisation or approximate normalisation with various grid sizes ($K$) in each subpanel.
Horizontal dashed lines show the prior density of a $\operatorname{Beta}(\eta = 1, \nu = 1)$.
In (c) we show the marginal posteriors for $\mu$ and $\tau$ under no normalisation, exact normalisation or approximate normalisation with various grid sizes ($K$) in each subpanel.
}
\label{fig:gaussian_results}
\end{figure}

We also evaluate the performance of adaptively building the grid of $a_0$ by comparing the mean absolute error (MAD) and root mean squared error (RMSE) of the estimated function $g_{\hat{\xi}}$ to the true (exact) normalisation when using either a uniform grid or the adaptive grid.
Over the whole range of $a_0 \in [0, 10]$ we found that the uniform grid leads to an estimated function with lower MAD ($0.59$ vs $0.82$) and lower RMSE ($0.99$ vs $1.18$). 
When considering only the range $a_0 \in [0, 1]$, the support of the prior -- $\pi_A(a_0\mid \delta)$ --,  we find the opposite: the adaptive grid outperforms uniform with MAD $2.10$ vs $0.10$ and RMSE $2.73$ vs $0.13$.
This suggests that the adaptive scheme would produce better results in situations where the region where the derivative changes lies within the support of the prior.

\subsubsection{Linear regression with a normal inverse-Gamma prior}
\label{sec:linreg_ex}

To conclude the examples for which we know $c(a_0)$ in closed-form, we present a popular model for Bayesian linear regression.
Suppose $\boldsymbol X_0$ is a $N_0 \times P$ full-rank matrix of predictors and $\boldsymbol y_0 = \{y_{01}, \ldots, y_{0N_0} \}$ is a vector of observations.
For illustrative purposes, we will employ a mean and variance parametrisation, which naturally leads to a normal inverse-Gamma conjugate prior.
The model is 
\begin{align*}
 \sigma^2 &\sim \operatorname{Inverse-Gamma}(\alpha_0, \gamma_0),\\
 \epsilon_i \mid \sigma^2  &\sim \operatorname{Normal}(0, \sigma^2), \\
 \beta \mid \sigma^2 &\sim \operatorname{Normal}(\boldsymbol \mu_0, \sigma^2\boldsymbol\Lambda_0^{-1}),\\
 y_{0i} &= \boldsymbol X_{0i}^\top \boldsymbol\beta + \epsilon_i, \\
\end{align*} 
where $\boldsymbol\beta$ is a $ 1 \times P$ vector of coefficients and $\boldsymbol\Lambda_0$ is a $P \times P$ variance-covariance matrix controlling the prior variance of the coefficients.
The posterior is again a normal inverse Gamma and thus
\begin{align}
 \label{eq:cA0_regression}
c(a_0) &= \sqrt{\frac{|\boldsymbol\Lambda_n|}{|\boldsymbol\Lambda_0^{-1}|}} \frac{\gamma_0^{\alpha_0}}{\gamma_n^{\alpha_n}}\frac{\Gamma(\alpha_0)}{\Gamma(\alpha_n)}  (2\pi)^{-N_0 a_0/2},\\
\nonumber
\boldsymbol\Lambda_n &= \boldsymbol X_{\star}^\top\boldsymbol X_{\star} + \boldsymbol \Lambda_0^{-1}, \\
\nonumber
\boldsymbol\mu_n &= \boldsymbol\Lambda_n^{-1}\left(\boldsymbol\Lambda_0^{-1}\boldsymbol\mu_0 + \boldsymbol X_{\star}^\top\boldsymbol y_{\star} \right),  \\
\nonumber
\alpha_n &= \alpha_0 + \frac{1}{2}a_0N_0,\\
\nonumber
\gamma_n &= \gamma_0 + \frac{1}{2}\left( \boldsymbol y_{\star}^\top \boldsymbol y_{\star} + \boldsymbol \mu_0^\top \boldsymbol \Lambda_0^{-1} \boldsymbol \mu_0 - \boldsymbol\mu_n^\top \boldsymbol \Lambda_n \boldsymbol \mu_n  \right),
\end{align}
where $\boldsymbol X_{\star} = \sqrt{a_0} \boldsymbol X_0$ and $\boldsymbol y_{\star} = \sqrt{a_0} \boldsymbol y_0$, and $|A|$ denotes the determinant of $A$.

As a first experiment, we generate $N_0 = 1000$ data points, drawing the columns of $\boldsymbol X_0$ from a standard normal distribution and using $\boldsymbol \beta = \{ -1, 1, 0.5, -0.5\}$.
The response variable $Y_0$ is generated using a normal distribution with variance $\sigma^2 = 4$, i.e. $y_{0i} \sim \operatorname{Normal}(\boldsymbol\beta^T \boldsymbol X_{0i}, 4)$.
For the current data, we generate $N = 100$ points using the same data-generating process.
To complete the model specification we set $\alpha_0 = 1/2$, $\gamma_0 = 2$ and $\boldsymbol\Lambda_0 = \frac{3}{2}\boldsymbol I_P$, where $\boldsymbol I_P$ is the $P \times P$ identity matrix.

Results of the power prior analysis of this data are shown in Table~\ref{tab:results_NIGregression} and indicate that while parameter recovery is similar for the exactly normalised and approximately normalised posteriors, the approximate method does not recover the lower tail of the marginal posterior of $a_0$ well.
\begin{table}[!ht]
\label{tab:results_NIGregression}
\caption{\textbf{Parameter estimates for the linear regression example}.
We report the posterior mean and 95\% BCI for the regression parameters $\boldsymbol\beta$, response variance $\sigma^2$ and the power prior scalar, $a_0$.
We employed a  $\operatorname{Beta}(\eta = \nu = 1)$ as prior for $a_0$.
}
{\small
\begin{tabular}{cccccc}
\hline 
 Parameter & True & None & Exact & app. $K = 50$ & app. $K = 10000$ \\
 \hline
$\beta_0$ & -1 & -0.56 (-0.98, -0.16) & -0.91 (-1.12, -0.59) & -0.92 (-1.12, -0.63) & -0.92 (-1.12, -0.68) \\
$\beta_1$ & 1 & 0.78 (0.38, 1.18) & 0.89 (0.66, 1.09) & 0.88 (0.69, 1.08) & 0.89 (0.67, 1.08) \\
$\beta_2$ & 0.5 & 0.32 (-0.05, 0.70) & 0.50 (0.24, 0.70) & 0.50 (0.28, 0.70) & 0.51 (0.29, 0.69) \\
$\beta_3$ & -0.5 & -0.68 (-1.04, -0.34) & -0.57 (-0.79, -0.39) & -0.58 (-0.78, -0.39) & -0.57 (-0.78, -0.38) \\
$\sigma^2$ & 4 & 3.7 (2.8, 4.8) & 4.4 (3.6, 5.0) & 4.4 (3.8, 5.1) & 4.4 (3.8, 5.0) \\
$a_0$ & -- & 0.00 (0.00, 0.00) & 0.48 (0.05, 0.97) &0.43 (0.11, 0.94) & 0.49 (0.09, 0.97)\\
\hline
\end{tabular}
}
\end{table}

Next, we explore the behaviour of our approach when the dimension of the problem increases, with the goal of ascertaining if and how the performance of the method deteriorates with increasing dimension.
We devised four scenarios where we keep constant the ratio $N_0/P = 10$ and make $P = 5, 10, 50, 100$ (see Table~\ref{tab:results_NIGregression_scenarios}).
For the current data, we fixed $N_0 = 100$, which leads to a near-identification configuration for Scenario D.
For these experiments we replaced the default configurations on Stan by increasing the number of iterations (from $2000$ to $5000$) and maximum tree size (\verb|max_treedepth| from $10$ to $15$) and decreasing the step size (\verb|adapt_delta| from $0.8$ to $0.95$).

The results in Table~\ref{tab:results_NIGregression_scenarios}  suggest that even in the extreme case of scenario D, with $P=100$ parameters and $N_0 = 1000$ data points our approach is able to accurately approximate the normalising constant and the approximately normalised posteriors compare favourably to their exactly normalised counterparts.
Perhaps counterintuitively, the MRAE in the estimation of the (log) normalising constant decreases with dimension.
We hypothesise this is the effect of the function $l(a_0)$ increasing in absolute value whilst the estimation method (bridge sampling) does not lose precision at quite the same rate, leading to relatively more precise estimates for values of $a_0$ closer to $1$.
In general, failing to account for the normalising constant leads to broader credibility intervals and worse estimates of the coefficients (using the marginal posterior mean) in terms of MSE.

\begin{table}[!ht]
\label{tab:results_NIGregression_scenarios}
\caption{\textbf{Scaling of the algorithm with dimension, linear regression example.}.
For each scenario we show the mean relative absolute error (MRAE) of the estimated $l(a_0)$ for $J = 20$ points.
We show the average width of the (95\%) credibility intervals (CIs) as well as the CIs that included the true data-generating coefficients (``inclusion'') for the unnormalised, approximately normalised and exactly normalised posteriors.
We also show the mean squared error (MSE) in the estimation of $\boldsymbol{\beta}$.
For comparison, the MRAE for $N_0 = 1000$ and $P =5$ was $0.05 \times 10^{-4}$.
}
\begin{center}
{\scriptsize
\begin{tabular}{cccccc}
\hline \\
                                                              &                                  & \multicolumn{4}{c}{Scenario}                                                          \\ 
                                                              &                                  & A                 & B                   & C                   & D                     \\ 
                                                              &                                Normalisation  & $N_0 = 50, P = 5$ & $N_0 = 100, P = 10$ & $N_0 = 500, P = 50$ & $N_0 = 1000, P = 100$ \\ \hline
\multirow{3}{*}{CI width (inclusion)}                         & None                     & 0.79 (0.8)        & 0.64 (1)            & 0.89 (0.94)         & 1.26 (0.87)           \\
                                                              & Approximate                      & 0.67 (0.6)        & 0.49 (0.9)          & 0.39 (0.98)         & 0.28 (1)              \\
                                                              & Exact                       & 0.67 (0.6)        & 0.49 (0.9)          & 0.38 (0.98)         & 0.28 (1)              \\
\multirow{3}{*}{MSE $\boldsymbol{\beta}$  ($\times 10^{-2}$)} & None                     & 6                 & 3                   & 5                   & 1.6                   \\
                                                              & Approximate                      & 5                 & 1.62                & 0.8                 & 0.31                  \\
                                                              & Exact                       & 5                 & 1.6                 & 0.8                 & 0.31 \\ 
                                                               MRAE $l(a_0)$ ($\times 10^{-4}$) & -- & 7.7               & 0.87                & 0.73                & 0.52                  \\\hline                
\end{tabular}
}
\end{center}
\end{table}

We present the estimated $l(a_0)$ in each scenario in Figure~\ref{sfig:ca0_NIGRegression} and show that the derivative-based method discussed briefly in Section~\ref{sec:adapt_grid_derivOnly} performs worse as the dimension of the problem increases, as expected.
This is because it gets progressively harder to reliably estimate the derivative of $l(a_0)$ as the dimension of the parameter space increases.
As a general takeaway we note that while the method remains accurate for this admittedly simple but high-dimensional problem, we needed to change the computational specifications to increase precision (e.g. increase the number of iterations) and also use a finer approximation grid for $l(a_0)$, in particular, $K = 50, 000$ and $K = 100, 000$ points.
For scenario D, even using $K = 50, 000$ did not lead to a good approximation (Figure~\ref{sfig:a0_posterior_NIGRegression_scenarios}D).

\subsection{Logistic regression}
\label{sec:logistic_regression}

Next, we approach a problem for which $c(a_0)$ cannot be written in closed-form.
Logistic regression is very popular model for binary outcomes in the presence of explanatory variables (covariates). 
Taking $\boldsymbol Y_0 = \{y_{01}, y_{02}, \ldots, y_{0N_0} \}$ with $y_{0i} \in \{0, 1\}$ and a (assumed full rank) $N_0 \times P$ matrix of covariates $\boldsymbol X_0$ as historical data, the model we consider here is 
\begin{align*}
 y_{0i} &\sim \operatorname{Bernoulli}(\theta_i), \\
 \theta_i &= \frac{\exp(\alpha + \boldsymbol X_{0i}^T \boldsymbol \beta)}{1 + \exp(\alpha + \boldsymbol X_{0i}^T \boldsymbol \beta)},\\
 \alpha & \sim \operatorname{Normal}(0, 1), \\
 \beta_i &\sim \operatorname{Normal}(0, 1),
\end{align*}
where $\alpha$ is the intercept and $\boldsymbol\beta$ is a $P$-dimensional vector of coefficients.
Since we do not have the benefit of a closed-form $c(a_0)$ in this example, we simulate data with known parameters and study how parameters are recovered as a function of the grid size $K$.
First, we generate $N_0 = 1000$ historical data points $(\boldsymbol Y_0, \boldsymbol X_0)$, where the matrix of covariates is constructed in the same manner as in the linear regression example.
We set $\alpha = 1.2$ and $\boldsymbol\beta = \{ -1, 1, 0.5, -0.5\}$.
For the current data, we use the same data-generating process to create a set of $N = 100$ new data points $(\boldsymbol Y, \boldsymbol X)$.
A prior sensitivity analysis is shown in Figure~\ref{sfig:sensitivity_logistic_regression}.
The chief idea is that a properly normalised power prior would allow one to capture the similarities between the historical and current data, while an analysis lacking the proper normalisation would yield counter-intuitive and suboptimal results, as demonstrated in the previous examples.
The results shown in Figure~\ref{fig:logistic_regression}a seem to support this intuition, since the approximately normalised power prior leads to posterior estimates that better recover the generating parameters, while the unnormalised prior leads to more diffuse posteriors that do not capture the full information contained in the data. 

\begin{figure}[!ht]
\hfill
\subfigure[Coefficients]{\includegraphics[width=7cm]{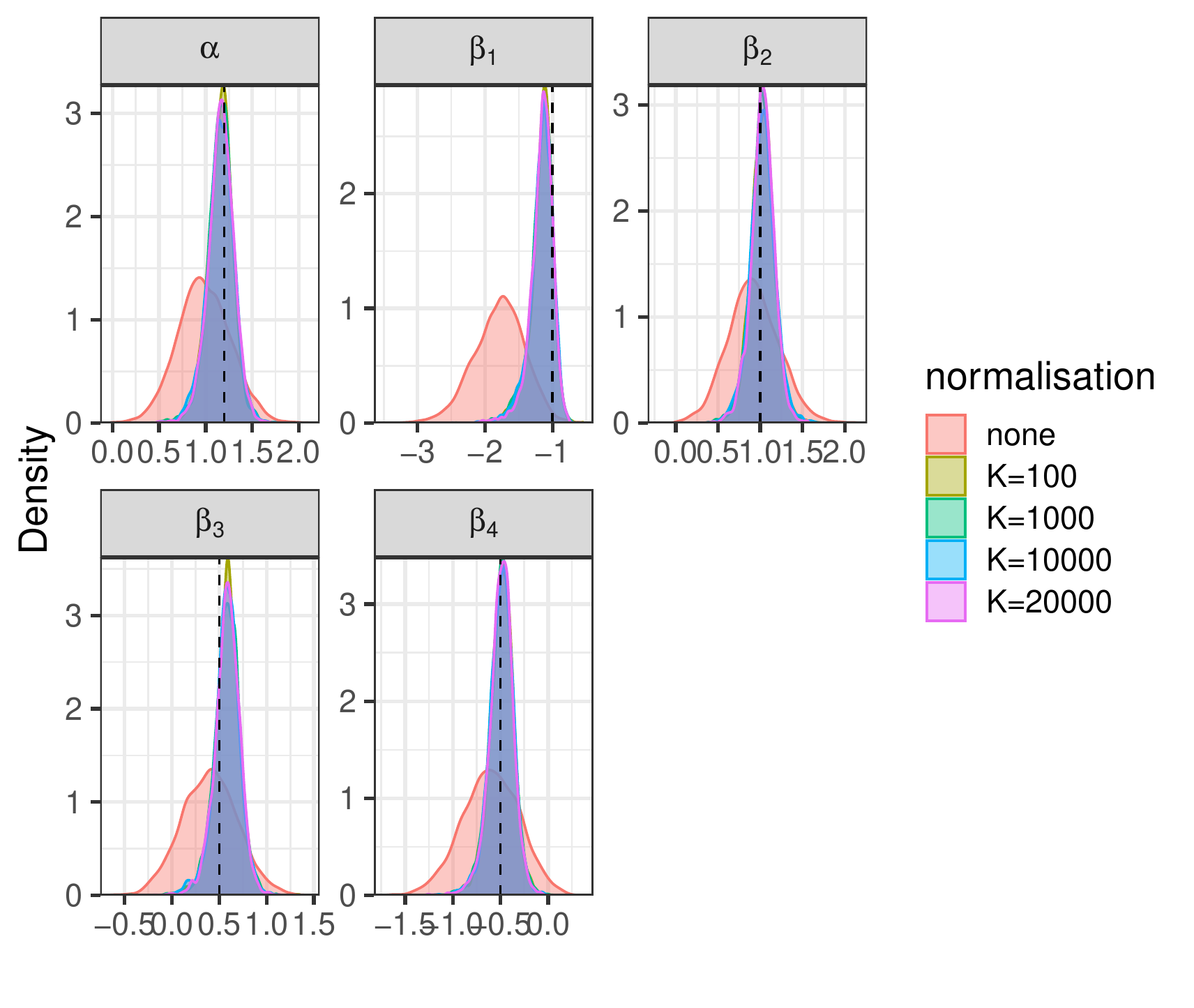}}
\hfill
\subfigure[$a_0$]{\includegraphics[width=7cm]{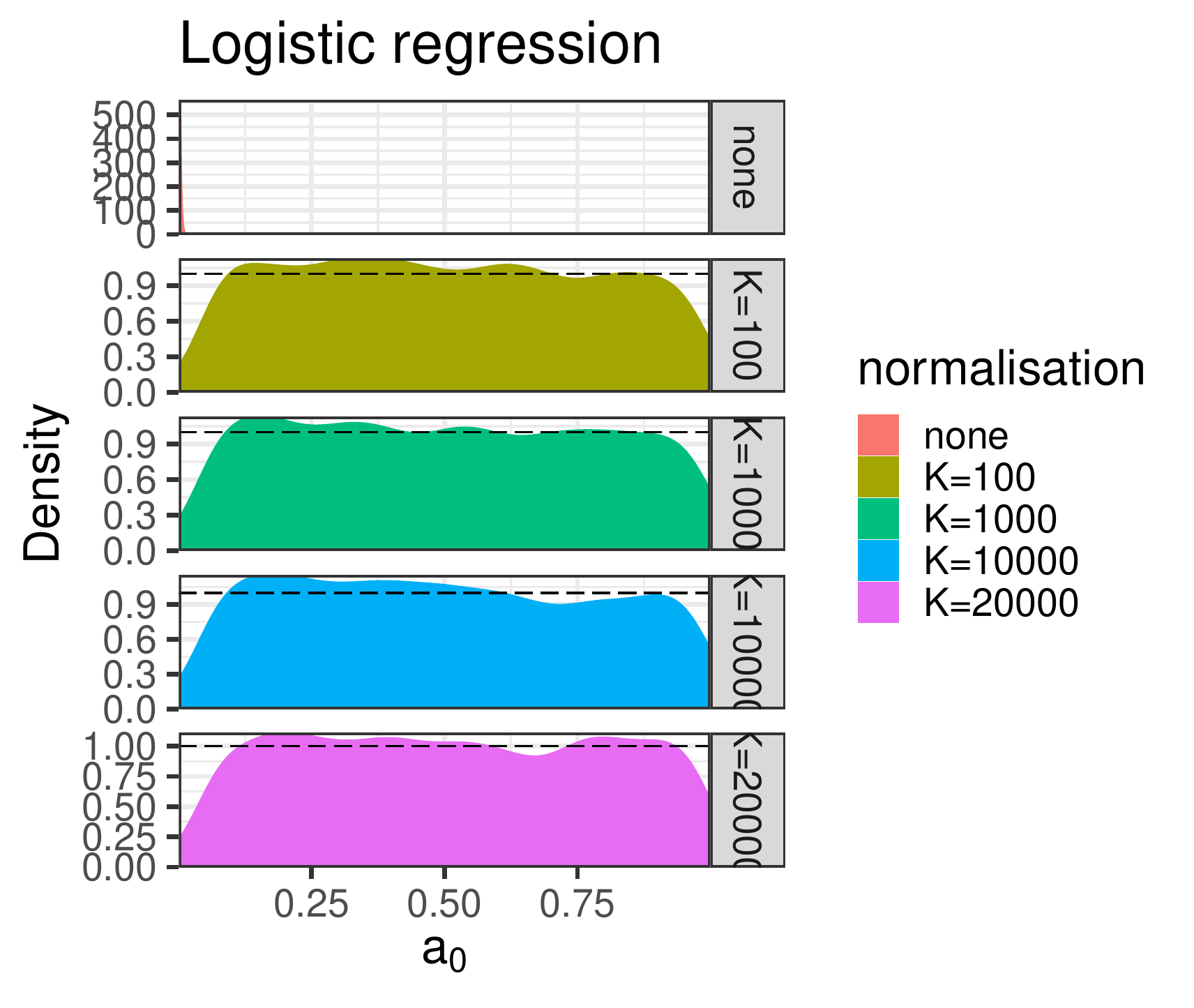}}
\hfill
\caption{\textbf{Results for the logistic regression example}.
Panel (a) shows the marginal posterior distributions for model parameters, with colours again pertaining to the approximation scheme.
Vertical dashed lines show the ``true'' parameter values of the data-generating process.
Horizontal dashed lines show the prior density of a $\operatorname{Beta}(\eta = 1, \nu = 1)$ for $a_0$.
In panel (b) the subpanels (and colours) correspond to the posterior distribution of the parameter $a_0$ when $c(a_0)$ is accounted for using various grid sizes $K$ and when it is not included.
}
\label{fig:logistic_regression}
\end{figure}

In addition we see that the approximate posteriors start to stabilise for $K > 1000$, showcasing the increased difficulty of this  multi-dimensional problem (see also Section~\ref{sec:linreg_ex}).
The bimodal marginal posterior for $a_0$ (Figure~\ref{fig:logistic_regression}b) suggests high uncertainty about the compatibility of the current data and historical data, which is unsurprising given the small number of current observations.
In order to gauge the dependence of these results on our specific parameter choices, in Figures~\ref{sfig:logistic_regression_extra}a and~\ref{sfig:logistic_regression_extra}b we show results for a similar setup with $\alpha = 0.2$ and $\boldsymbol\beta = \{ -10, 1, 5, -5\}$.
Results indicate that in this scenario with larger (absolute value) coefficients, our approximation still works well, albeit with less posterior coverage of the data-generating values.
Overall, the results of this section suggest that our approach works well in a setting where the normalising constant is not known in closed-form and leads to posterior estimates that appropriately incorporate the information in the historical data.

\subsection{Survival model with cure fraction}
\label{sec:survival}

As a final illustration, we show an application of the approximation scheme to an elaborate survival model, namely the cure rate model proposed by~\cite{Chen1999a}.
This model allows one to accommodate situations where a significant proportion of subjects is cured.
The model can be described generatively as follows.
Let $N$ be the number of carcinogenic cells left after initial treatment, assumed to follow a Poisson distribution with rate $\theta$.
Now let $Z_j$, $j = 1, 2, \ldots, N$ be i.i.d. random variables with distribution function $F(t) = 1 - S(t)$.
The variable of interest is then $T = \min(Z_j)$, $0 \leq j \leq N$, the time of relapse.
Suppose we observe i.i.d. data $\boldsymbol Y = \{ y_1, y_2, \ldots, y_n \}$ with $\boldsymbol W = \{ w_1, w_2, \ldots, w_n\}$ being indicators of whether observations are (right) censored.
If we also have a $n \times p$ matrix of covariates $\boldsymbol X$, we can then write the likelihood after marginalising over the latent variables:
\begin{equation}
\label{eq:cure_rate_likelihood}
L(\boldsymbol\beta, \boldsymbol \psi \mid \boldsymbol Y, \boldsymbol W) = \prod_{i = 1}^n \left(\theta_i f(y_i \mid \boldsymbol\phi) \right)^{w_i} \exp\left[ -\theta_i \left( 1- S(y_i \mid \boldsymbol\phi \right) \right],
\end{equation}
where $\theta_i = \exp(\boldsymbol X_i^T\boldsymbol\beta)$, with $\boldsymbol \beta$ a vector of coefficients and $\boldsymbol \psi = (\alpha, \lambda) $ the parameters of a Weibull distribution, i.e.,
$$ f(y_i \mid \boldsymbol \psi) = \alpha y_i^{\alpha -1} \exp\left[ \lambda - y_i^\alpha \exp(\lambda)\right].$$
For more details, see~\cite{Chen1999a}.

To illustrate the use of a normalised power prior, we will consider a situation where one wants to analyse data from a current clinical trial in light of historical information provided by an earlier study which includes many of the same covariates and measurements.
In particular, we consider data from a two-arm clinical trial on phase III melanoma conducted by the Eastern Cooperative Oncology Group, denoted E1684.
In this study, patients were assigned to either a interferon treatment (IFN) or observation, and survival was defined as time from randomisation to death.
We have $n = 284$ measurements for this data set.
As historical data, we employ the data from an earlier essay, denoted E1673, for which we have $n_0 = 650$ data points.
We consider three covariates: (standardised) age, sex and performance status (PS), i.e., whether the patient was fully active or other.
Since in this paper we are chiefly concerned with situations where the initial priors are proper, we modify the prior modelling of~\cite{Chen1999a} to include proper priors for all parameters.
Namely, we employ the following prior structure:
\begin{align*}
 \boldsymbol \beta &\sim \operatorname{Normal}(\boldsymbol{0}, \sigma_\beta^2 \boldsymbol{I}_P), \\
 \alpha & \sim \operatorname{Gamma}(\delta_0, \tau_0),\\
 \lambda  & \sim \operatorname{Normal}(\mu_0, \sigma_0^2),
\end{align*}
with $\sigma_\beta^2 = 10$, $\delta_0  = 1$, $\tau_0 = 0.01$, $\mu_0 = 0$ and $\sigma_0^2 = 10, 000$.
Letting $D_0 = \{ \boldsymbol Y_0, \boldsymbol W_0, \boldsymbol X_0\}$, the normalised joint power prior and posterior are, respectively,
\begin{equation}
 \label{eq:power_prior_cure_rate}
 \pi(\boldsymbol \beta, \boldsymbol \psi, a_0 \mid D_0) = \frac{L(\boldsymbol\beta, \boldsymbol \psi \mid D_0)^{a_0}\pi(\boldsymbol \beta, \boldsymbol \psi) \pi_A(a_0)}{c(a_0)},
\end{equation}
and
\begin{equation}
 \label{eq:power_posterior_cure_rate}
 p(\boldsymbol \beta, \boldsymbol \psi, a_0 \mid D_0, D) \propto \frac{L(\boldsymbol\beta, \boldsymbol \psi \mid D_0)^{a_0}\pi(\boldsymbol \beta, \boldsymbol \psi) L(\boldsymbol\beta, \boldsymbol \psi \mid D) \pi_A(a_0) }{c(a_0)},
\end{equation}
where $D = \{ \boldsymbol Y, \boldsymbol W, \boldsymbol X\}$ and we take $\pi_A(\cdot)$ to be a Beta prior with parameters $\eta$ and $\nu$, taking values as in Table~\ref{tab:survival_results}.
In their original analysis of the melanoma data, \cite{Chen1999a} employed an unnormalised power prior.
Here, we revisit their analysis (Table 4 therein) and employ a normalised power prior which can then be compared to the results with an unormalised prior (Table~\ref{tab:survival_results}).
The sensitivity analysis presented in Figure~\ref{sfig:cure_rate_sensitivity} suggests that all model parameters are very sensitive to small values of $a_0$, whereas posteriors stabilise for values of $a_0 > 0.1$.

The results in Table~\ref{tab:survival_results} show that including the normalisation factor $c(a_0)$ leads to substantially different parameter estimates.
In particular, the normalised prior leads to the posterior BCI for the coefficient of age and sex excluding zero, showing that when information is properly accounted for through correct normalisation of the power prior, inferences might change.
Looking at the posterior estimates for $a_0$, we note that these changes likely stem from the fact that when employing a (approximately) normalised power prior, we give the historical data more weight and thus effectively increase the amount of data entering the model.

\begin{table}[!ht]
\caption{\textbf{Results for the cure fraction rate model}.
For several choices of prior for $a_0$, we show posterior means and 95\% BCIs for the model coefficients as well as $\boldsymbol \psi = \{ \alpha, \lambda\}$ and $a_0$ under no normalisation and approximate normalisation using the methods proposed in Section~\ref{sec:efficient_computation_ca0}.
We employed $J = 20$ evaluations to estimate $c(a_0)$ and $K = 2E4$ points for the approximation grid.
}
\begin{tabular}{cccc}
\hline
                                            &           & \multicolumn{2}{c}{}                      \\ \hline
Prior on $a_0$, $\operatorname{Beta}(\eta, \nu)$    & Parameter & Unnormalised         & App. Normalised  \\
\hline
\multirow{7}{*}{$\eta = 1$, $\nu  = 1$}     & Intercept & 0.10 (-0.11, 0.31)   & 0.47 (0.24. 0.70)           \\
                                            & Age       & 0.09 (-0.05,0.23)    & 0.15 (0.03, 0.26)           \\
                                            & Sex    & -0.13 (-0.44, 0.18)  & -0.31 (-0.48, -0.13)        \\
                                            & PS        & -0.23 (-0.76, 0.25)  & -0.04 (-0.33, 0.36)         \\
                                            & $\alpha$  & 1.30 (1.13, 1.48)    & 1.02 (0.93, 1.12)           \\
                                            & $\lambda$ & -1.36 (-1.63, -1.12) & -1.80 (-2.03, -1.55)        \\
                                            & $a_0$     & 0.00 (0.00, 0.00)    & 0.41 (0.19, 0.94)           \\
\multirow{7}{*}{$\eta = 50$, $\nu  = 50$}   & Intercept & 0.20 (-0.02, 0.44)   & 0.50 (0.29, 0.71)           \\
                                            & Age       & 0.10 (-0.04, 0.24)   & 0.15 (0.05, 0.26)           \\
                                            & Sex    & -0.17 (-0.44, 0.10)  & -0.33 (-0.49, -0.19)         \\
                                            & PS        & -0.19 (-0.72, 0.27)  & 0.07 (-0.25, 0.37)          \\
                                            & $\alpha$  & 1.17 (1.02, 1.32)    & 1.01 (0.92, 1.10)           \\
                                            & $\lambda$ & -1.47 (-1.75, -1.21) & -1.83 (-2.05, -1.62)        \\
                                            & $a_0$     & 0.03 (0.02, 0.04)    & 0.48 (0.37, 0.60)           \\
\multirow{7}{*}{$\eta = 100$, $\nu  = 100$} & Intercept & 0.28 (0.05. 0.52)    & 0.51 (0.29, 0.72)           \\
                                            & Age       & 0.11 (-0.02, 0.16)   & 0.16 (0.05, 0.26)           \\
                                            & Sex    & -0.20 (-0.43, 0.02)  & -0.33 (-0.49, -0.18)        \\
                                            & PS        & -0.15 (-0.61, 0.27)  & 0.07 (-0.27, 0.37)          \\
                                            & $\alpha$  & 1.11 (0.98, 1.24)    & 1.01 (0.92, 1.09)           \\
                                            & $\lambda$ & -1.56 (-1.84, -1.30) & -1.83 (-2.05, -1.63)        \\
                                            & $a_0$     & 0.07 (0.05, 0.08)    & 0.49 (0.42, 0.56)           \\
\multirow{7}{*}{$\eta = 200$, $\nu  = 1$}   & Intercept & 0.38 (0.14, 0.63)    & 0.54 (0.36, 0.72)           \\
                                            & Age       & 0.13 (0.01, 0.25)    & 0.17 (0.09, 0.26)           \\
                                            & Sex    & -0.25 (-0.45, 0.06)  & -0.36 (-0.49, -0.24)        \\
                                            & PS        & -0.09 (-0.52, 0.31)  & 0.15 (-0.11, 0.39)          \\
                                            & $\alpha$  & 1.05 (0.93, 1.17)    & 1.00 (0.93, 1.07)           \\
                                            & $\lambda$ & -1.68 (-1.96, -1.48) & -1.89 (-2.06, -1.73)        \\
                                            & $a_0$     & 0.14 (0.12. 0.16)    & 1.00 (0.98, 1.00)           \\
\hline                                            
\end{tabular}
\label{tab:survival_results}
\end{table}

\section{Discussion}
\label{sec:discussion}

\subsection{Starting from a sensitivity analysis}
\label{sec:sensitivity_analysis}

The starting point for the methodology presented here is a prior sensitivity analysis (PSA), in which one computes the distribution $L(D_0 \mid \theta)^{a_0}\pi(\theta)$ for a range of values of $a_0$ in order to gauge how sensitive the resulting prior is to the discounting (tempering) parameter. 
The class of models amenable to such an analysis thus comprises models: (i) that are well-established/studied and thus there is little need to test different likelihood functions or even initial priors; and (ii) for which one is able to compute the estimates in reasonable time such that a sensitivity analysis of the sort discussed here is feasible. 

Taking these conditions as given, we then first propose a simple way of picking a fixed budget of, say, $J=20$, values for $a_0$ at which to compute the power prior distribution using a bisection-type algorithm based on the theoretical results in Section~\ref{sec:properties}.
Since there are many instances in which one would wish to represent the uncertainty about $a_0$ as probability distribution, we propose a way to recycle computations in order to approximately sample from the joint posterior of $(a_0, \boldsymbol{\theta})$ where $\boldsymbol{\theta}$ are the parameters of interest.
This requires computing the normalising constant (Eq~\ref{eq:normconst}) as observed by~\cite{Neuenschwander2009}.

Sensitivity to the effects of normalisation varies between models and data configurations; the Bernoulli model in Section~\ref{sec:reproduce_N2009} shows little difference in parameter estimates between unormalised and normalised posteriors, whilst for the regression and survival examples -- Sections~\ref{sec:linreg_ex} and~\ref{sec:survival}, respectively -- parameter estimates, in particular their precisions are affected more strongly.
In terms of shape, the normalising constant $c(a_0)$ seen as function of the discounting scalar $a_0$ is usually monotonic, at least for majority of the examples we have considered.
The notable exception is the somewhat artificial example of a Gaussian likelihood in Section~\ref{sec:gaussian_illus} (Figure~\ref{fig:gaussian_results}a) for which $c(a_0)$ resembles a convex parabola, illustrating the results in Section~\ref{sec:properties}, which tell us that the normalising constant is a strictly convex function of the discounting scalar, $a_0$.
This problem was devised so as to test our ability to approximate $c(a_0)$ in a difficult setting, namely when it is not monotonic and varies over a large range -- we give more motivation for a theoretical analysis in Appendix~\ref{sec:ca0_norm_deriv}.
Results indicate the method proposed here is able to correctly approximate the normalising constant and thus provide a usable technique when $c(a_0)$ is not known in closed-form.

\subsection{The normalised power prior as a doubly-intractable problem}
\label{sec:doubly_intractable}

The normalised power prior is closely related to the class of doubly-intractable problems, which encompasses Markov random fields~\citep{Besag1974} and exponential random graph models~\citep{Robins2007} and many others.
For a review, see~\cite{Park2018}.

To see how our problem fits into the doubly-intractable framework, we can re-write equation~(\ref{eq:joint_posterior}) as  
\begin{align*}
 p(a_0, \theta \mid D_0, D, \delta) &\propto L(D \mid \theta, a_0) \pi(\theta, a_0 \mid D_0, \delta), \\
 &\propto L(D \mid \theta, a_0) h(\theta \mid a_0, D_0) \pi_A(a_0 \mid \delta),\\
\end{align*}
with $h(\theta \mid a_0, D_0) := c(a_0)^{-1} L(\theta \mid D_0)^{a_0}\pi(\theta)$ playing the part of an intractable likelihood where $\theta$ is seen as data.
With the exception of the double Metropolis-Hastings algorithm of~\cite{Liang2010}, most computational tools available rely on the ability to simulate from $h(\theta \mid a_0, D_0)$ relatively easily, which is often not the case in our setting (see below). 

On the other hand, the fact that $c(a_0)$ is univariate allows our approximation scheme to be feasible for many models. 
In contrast, extending our approach to multiple historical data sets (Remark~\ref{rmk:historical}) would thus be a non-trivial task, since one would need to ``observe'' $c(a_{01}, a_{02}, \ldots, a_{0M})$ at many points (on a $M$-dimensional grid) in order to obtain a good approximation.

\subsection{Exact and inefficient or inexact and efficient?}

In contrast to many existing algorithms such as auxiliary variable MCMC~\citep{Moller2006}, the methodology we put forth in this paper does not lead to sampling from the exact joint posterior of $a_0$ and $\theta$.
Our method is what~\cite{Park2018} call an ``asymptotically inexact'' algorithm because we replace the true (power) prior with an approximate density with normalising constant $g_{\boldsymbol{\hat{\xi}}}(a_0)$.

While it would obviously be preferable to have an exact algorithm, it is important to strike a balance between simplicity and exactitude.
The noise in an inexact algorithm can be decomposed into approximation error and Monte Carlo error, whereas the noise an exact algorithm comes solely from the Monte Carlo approximation and thus can, in theory, be made arbitrarily small.
In practice, however, it is entirely possible for the error from a suboptimally implemented exact algorithm to be larger than that of an efficient inexact method.
Almost all available state-of-the-art exact samplers for doubly intractable problems require careful consideration of the proposal distributions, as in the case of the double Metropolis-Hastings sampler of~\cite{Liang2010}, and/or the ability to easily simulate from the intractable likelihood~\citep{Murray2012,Park2018,Stoehr2019}, which in our case is not feasible.

Here we have devised a simple framework that employs the very efficient dynamic Hamiltonian Monte Carlo implemented in Stan~\citep{Carpenter2017} and requires very little programming effort to include any model from the class discussed above.
Our results show that the adaptive grid-building with GAM-based approximation works well for a range of problems and this gives us confidence that in this instance one should prefer an efficient inexact algorithm to a potentially inefficient exact one.

\subsection{Current limitations and future directions}
\label{sec:future}

The method presented here can be improved in many respects.
First, it is possible that better approximations  to $c(a_0)$ could be devised by using custom curve-fitting methods that incorporate the fact that  $l^\prime(a_0)$ -- and  $l^\prime(a_0)$ -- is monotonically increasing (see Section~\ref{sec:properties}), such as the Gaussian process methods discussed by~\cite{Riihimaki2010} and~\cite{Wang2016}.

Secondly, extending the methodology here to multiple historical data sets is straightforward only under the assumption of independence between data sets.
The grid-based approach that we have shown to work well here is going to scale poorly with dimension in the sense that if one has $K$ historical data sets and decides to use $J$ points for the sensitivity analysis, one ends up computing $KJ$ posteriors.
Moreover, under non-independence incorporating uncertainty about multiple weights at once would necessitates careful consideration of the prior distribution over $\boldsymbol{a_0}$.

Finally, while our inexact approach makes it possible for practitioners to perform sensitivity analyses and sample from the approximate joint posterior efficiently, this should not discourage the development of more efficient exact algorithms.
The main challenge for the normalised power prior in particular is that it is not easy to sample from 
$h(\theta \mid a_0, D_0)$, making it difficult to implement auxiliary variable-type algorithms.
This points to double-Metropolis~\cite{Liang2010}-type algorithms as the most promising class of exact algorithms to be developed for the analysis of the normalised power prior.

\section*{Acknowledgements}

The authors would like to thank Aditya Ravuri for pointing out the first part of the proof of Theorem 1. 
LMC would like to thank Leo Bastos for helpful discussions, Dr. Beat Neuenschwander for clarifications regarding his paper and Chris Koenig and Ben Jones for testing the computer code developed for this paper.
This study was financed in part by the Coordenação de Aperfeiçoamento de Pessoal de Nível Superior - Brasil (CAPES) Finance Code 001.

\bibliography{power_prior}

\appendix

\counterwithin{figure}{section}
\counterwithin{table}{section}

\section{Additional results and proofs}
\label{sec:further_proofs}

Proof of Theorem~\ref{thm:integrability}.

\begin{proof}
Denote $f_{a_0}(D_0;\theta) := L(D_0 \mid \theta)\pi(\theta)$.
First, note that $c(0) = 1$ because $\pi$ is proper.
For $0 < a_0 \leq 1$ the function $g(x) = x^{a_0}$ is concave and thus, by Jensen's inequality and the finiteness of $L( D_0 \mid \theta)$ for all of its arguments we have
\[ c(a_0) = \int_{\Theta} f_{a_0}(D_0; \theta) \, \, d\theta \leq \left[ \int_{\Theta} L(D_0  \mid \theta)\pi(\theta) \, \, d\theta \right]^{a_0} < \infty. \]
Rewrite $f_{a_0}(D_0; \theta) = L(D_0 \mid \theta)^{a_0 -1} L(D_0 \mid \theta)\pi(\theta)$.
If $1 \leq a_0 \leq 2$, we have the Jensen's inequality case above, since we know that $L(D_0 \mid \theta)\pi(\theta)$ is normalisable (proper).
Similarly, if $2 \leq a_0 \leq 3$, we can write 
\[  f_{a_0}(D_0; \theta) = L(D_0 \mid \theta)^{a_0-p} L(D_0 \mid \theta)^p\pi(\theta), \]
with $1 \leq p \leq 2$, again falling into the same case, since we know that $L(D_0 \mid \theta)^{p}\pi(\theta)$ is normalisable.
We can then show that for any $n \in \mathbb{N}$, $\int_{\Theta}f_{a_0}( D_0 ; \theta)\, d\theta < \infty$ for  $n-1 \leq a_0 \leq n$.
The base case for $1 \leq n \leq 3$ is established.
Now suppose the hypothesis holds for $n \geq 3$.
For $ n \leq  a_0 \leq n + 1$ and $n-1 \leq p_n \leq n$:
\begin{align*}
 \int_{\Theta} L(D_0 \mid \theta)^{a_0-p_n} L(D_0 \mid \theta)^{p_n}\pi(\theta)\, d\theta < \infty, \\
\end{align*}
because $0 \leq a_0 - p_n \leq 1$ and $L(D_0 \mid \theta)^{p_n}\pi(\theta)$ is proper by hypothesis, establishing the case for $n + 1$.
\end{proof}

\begin{remark}
 \label{rmk:improper}
 \textbf{Improper initial priors}. If $\pi$ is improper but $L(\theta \mid D_0)\pi(\theta)$ is integrable, i.e. the posterior is proper, then Theorem~\ref{thm:integrability} holds for $a_0 > 0$.
\end{remark}
\begin{proof}
 Analogous to the proof of Theorem~\ref{thm:integrability}, only excluding the boundary case $a_0 = 0$.
\end{proof}

Now let us prove Remark~\ref{rmk:historical}:
\begin{proof}
 Recall that the power prior on multiple historical data sets is of the form~\citep[Eq. 2.9]{Ibrahim2015}:
 \[ \pi(\theta \mid \boldsymbol D, \boldsymbol a_0) \propto \prod_{k=1}^M L(\theta \mid D_k)^{a_{0k}} \pi_0(\theta). \]
Assume, without loss of generality, that $L(\theta \mid  D_k)^{a_{0k}} > 1$ for all $\theta$ and let $m := \max(\boldsymbol a_0)$ with $\boldsymbol a_0 := \{ a_{01}, a_{02}, \ldots, a_{0M}\}$.
Then $\pi(\theta \mid \boldsymbol D, a_0)$ is bounded above by 
\[  g(\theta) :=  \prod_{k=1}^M L(\theta \mid D_k)^{m} \pi_0(\theta) =  \left[ \prod_{k=1}^M L(\theta \mid D_k) \right]^m  \pi_0(\theta) = L(\theta \mid \boldsymbol D)^m \pi_0(\theta), \]
which is normalisable following Theorem~\ref{thm:integrability}.
To relax the assumption made in the beginning, notice that this construction also bounds the case $ 0 \leq  L(\theta \mid  D_k)^{a_{0k}} \leq 1$ (for some $k$) above.
\end{proof}

To prove Lemma~\ref{lm:convex_norm_constant}, it is convenient to first establish the following proposition:

\begin{proposition}
\label{prop:c_is_Cinfinity}
All of the derivatives of $c(a_0)$ exist, i.e., $c \in \mathcal{C}^{\infty}$.
\end{proposition}
\begin{proof}
First, we will assume that $L(D_0 \mid \theta) > 0\: \forall \theta \in \Theta$.
For convenience, let
\[  f(\theta) = \frac{L(D_0 \mid \theta)^{a_0} \pi(\theta)}{c(a_0)}. \]
Now, consider the change of variables $\theta \mapsto l$, with $l = \log(L(D \mid \theta))$.
Then we write
\[ h (l) = \frac{\exp(a_0 l) g(l)}{z(a_0)},\]
where $g(l)$ is a non-negative function that accommodates the transform $\theta \mapsto l$ with respect to the prior $\pi$ and $z(a_0)$ is the appropriate normalsing constant, guaranteed to exist by Theorem~\ref{thm:integrability}.
The moment-generating function (MGF) of $l$ is 
\[ M_t(l) = E_h[\exp(tl)] = \int_{-\infty}^\infty \frac{\exp((t + a_0) l) g(l)}{z(a_0)}\, dl.\]
Since $E_h[l^r] \equiv  \frac{d^r c(a_0)}{d a_0^r}$, all that remains is to show that $M_r(l)$ exists for all $r \geq 0$.
Under the change of variables discussed above, Theorem~\ref{thm:integrability} shows that
\[ \int_{-\infty}^\infty \exp(wl) g(l)\, dl < \infty,\]
for $w > 0$.
Making $w = t + a_0$ concludes the proof.
\end{proof}

Now we establish Lemma~\ref{lm:convex_norm_constant}.
\begin{proof}
Define the normalising constant as a function $c : [0, \infty) \to (0, \infty)$,
\begin{equation}
 \label{eq:normconst}
 c(a_0) := \int_{\Theta} L(D_0 \mid \theta)^{a_0} \pi(\theta)\, d\theta,
\end{equation}
which is positive and continuous on its domain.
The first and second derivatives are
\begin{align}
\label{eq:derivative_ca0}
c^\prime(a_0) &= \int_{\Theta} L(D_0 \mid \theta)^{a_0} \pi(\theta) \log L(D_0 \mid \theta) \, d\theta, \\
c^{\prime\prime}(a_0) &= \int_{\Theta} L(D_0 \mid \theta)^{a_0} \pi(\theta) [\log L(D_0 \mid \theta)]^2 \, d\theta,
\end{align}
and the integrals always exist (as per Proposition~\ref{prop:c_is_Cinfinity}).
Differentiation under the integral sign is justified because both $L(D_0 \mid \theta)^{a_0} \pi(\theta)$ and $L(D_0 \mid \theta)^{a_0} \pi(\theta) \log L(D_0 \mid \theta)$ are continuous with respect to $\theta$.
From this we conclude that $c$ is (strictly) convex and $c^\prime$ is monotonic, because $c^{\prime\prime}$ is always positive.
\end{proof}

\section{The derivative of $c(a_0)$ for the normal case}
\label{sec:ca0_norm_deriv}

In this section we give more detail on the analysis of the Gaussian example of Section~\ref{sec:gaussian_illus} in the main text.
Define
\begin{align*}
 c(a_0) &= g(a_0)h(a_0)w(a_0)z(a_0), \\
 g(a_0) &:=  \frac{\Gamma\left( \alpha_0 + \frac{N_0}{2}a_0 \right)}{\Gamma(\alpha_0)}, \\
 h(a_0) &:= \frac{\beta_0^{\alpha_0}}{ \left(  \beta_0 + \Delta a_0\right)^{\alpha_0 + \frac{N_0}{2}a_0}}, \\
 w(a_0) &:= \left(\frac{\kappa_0}{\kappa_0 + N_0 a_0} \right)^2 , \\
 z(a_0) &:= (2\pi)^{-N_0 a_0/2}, 
\end{align*}
with $\Delta =  \frac{1}{2}\left( \sum_{i=1}^{N_0}(y_{0i}-\bar{y})^2 + \frac{\kappa_0}{\kappa_n} N_0 (\bar{y}-\mu_0)^2 \right)$.
Thus, dropping dependency on $a_0$ for notational compactness, we have
\begin{equation}
\label{eq:c_deriv_gaussian}
 c^\prime = h w z g^\prime + g w z h^\prime + g h z w^\prime + g h w z^\prime.
\end{equation}
Notice that only the first term of~(\ref{eq:c_deriv_gaussian}) is positive.
Since $g^\prime(a_0) = \frac{N_0}{2} \psi_0\left( \alpha_0 +  \frac{N_0}{2} a_0 \right)g(a_0)$, we can write the following inequality:
\begin{equation*}
 c^\prime(a_0) > 0 \implies \frac{N_0}{2} \psi_0\left( \alpha_0 +  \frac{N_0}{2} a_0 \right) > \frac{|h^\prime(a_0)|}{h(a_0)} + \frac{|w^\prime(a_0)|}{w(a_0)} + \frac{|z^\prime(a_0)|}{z(a_0)}.
\end{equation*}
Since
\begin{align*}
 \frac{|h^\prime(a_0)|}{h(a_0)}  &=  \frac{\Delta\left( \alpha_0 + \frac{N_0}{2} a_0 \right) }{\Delta a_0 + \beta_0} + \frac{N_0}{2}\log{\left( \Delta a_0+ \beta_0 \right) }, \\
 \frac{|w^\prime(a_0)|}{w(a_0)}  &= \frac{2N_0}{a_0N_0+\kappa_0},\\
\frac{|z^\prime(a_0)|}{z(a_0)}  &= \log(2\pi) \frac{N_0}{2},
\end{align*}
we arrive at
\begin{align}
\nonumber
 \frac{N_0}{2} \psi_0\left( \alpha_0 +  \frac{N_0}{2} a_0 \right) &>  \frac{\Delta\left( \alpha_0 + \frac{N_0}{2} a_0 \right) }{\Delta a_0 + \beta_0} + \frac{N_0}{2}\log{\left( \Delta a_0+ \beta_0 \right) } + \frac{2N_0}{a_0N_0+\kappa_0} + \log(2\pi) \frac{N_0}{2}, \\
\psi_0\left( \alpha_0 +  \frac{N_0}{2} a_0 \right) &>  \frac{\Delta\left( 2\alpha_0 + N_0a_0 \right) }{N_0\left(\Delta a_0 + \beta_0\right)} + \log{\left( \Delta a_0+ \beta_0 \right) } + \frac{4}{a_0N_0+\kappa_0} + \log(2\pi).
\end{align}


\section{Comparing approximations of $l(a_0)$}
\label{sec:derivative_only}

In this section we study two approaches to estimating $l(a_0)$, using four examples where it is known in closed-form.
First, we consider the main approach discussed in this paper, which consists of estimating $l(a_0)$ at a grid of $J = 15$ points of $a_0$ and using a GAM as the approximating function $g_\xi$ to approximate $l(a_0)$ directly.
We then evaluate the fitted function at a grid of $K = 20,000$ values to form a vector $\boldsymbol l_{\text{direct}}$.

Another approach is to use estimates of $l^\prime(a_0)$ (see Equation~\ref{eq:derivative_ca0}) as data and fit a GAM as the approximating function $h_\omega$ and evaluate this function at a fine grid of values for $a_0$.
We can then approximate $l(a_0)$ by midpoint integration, forming a vector of predictions/estimates $\boldsymbol l_{\text{deriv}}$.
Other methods, such as trapezoid integration might also be used.
We then compare the estimated values with the true values, $\boldsymbol l_{\text{true}}$, by computing the root mean squared error, $\hat{r} = \sqrt{\frac{1}{K} \sum_{i= 1}^K \left( l^{(i)}_{\text{est}} - l^{(i)}_{\text{true}} \right)^2 }$.

We show results for the Bernoulli, Poisson, Gaussian and linear regression in Table~\ref{tab:rmse_approx}.
Results are presented for two values of the $a_0$ endpoint, $M = 1$ and $M = 10$.
As expected, estimates (predictions) derived using direct estimation of $l(a_0)$ are substantially more accurate.
The only instance in which the derivative-based method is more accurate is for the Gaussian model and for a large endpoint $M = 10$.
Of the four models considered, only the Gaussian example (see section~\ref{sec:gaussian_illus}) has a non-monotonic $l(a_0)$, which might explain the observed results.

\begin{table}[!ht]
\caption{\textbf{ Mean root squared error comparison of methods for approximating $l(a_0)$}.
We used $J = 20$ points to construct $\boldsymbol a^{\text{est}}$ and use a GAM to approximate either $l(a_0)$ or $l^\prime(a_0)$. 
In the latter case, we evaluate the fitted function on a fine grid ($K = 20, 000$ points) and obtain an approximation of $l(a_0)$~\textit{via} midpoint integration (see text).
}
\begin{center}
\label{tab:rmse_approx}
\begin{tabular}{ccccc}
\hline
        Model          & \multicolumn{2}{c}{$M = 1$} & \multicolumn{2}{c}{$M = 10$} \\
\hline
                  & Direct & Deriv + midpoint & Direct  & Deriv + midpoint \\
Bernoulli         & 0.08   & 1.61             & 0.17    & 1.03             \\
Poisson           & 0.05   & 1.02             & 0.13    & 1.97             \\
Linear regression & 0.33   & 6.33             & 0.71    & 7.06             \\
Gaussian          & 1.74   & 10.08            & 31.29   & 12.11           \\
\hline
\end{tabular} 
\end{center}
\end{table}

\newpage

\section{Supplementary Figures}
\label{sec:extra_figs}

\begin{figure}[!ht]
\begin{center}
 \hfill
\subfigure[Scenario 1]{\includegraphics[width=7cm]{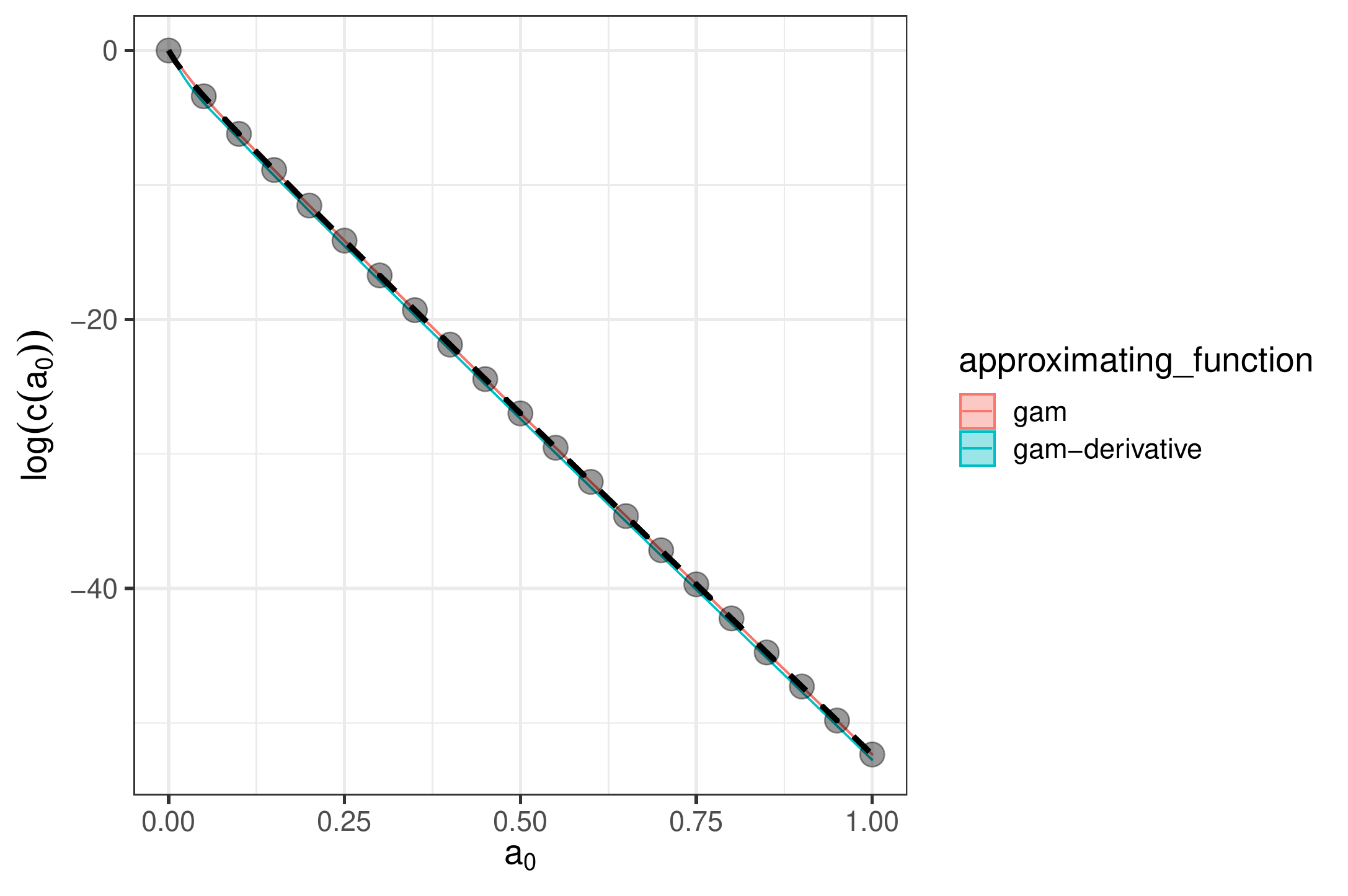}}
\hfill
\subfigure[Scenario 2]{\includegraphics[width=7cm]{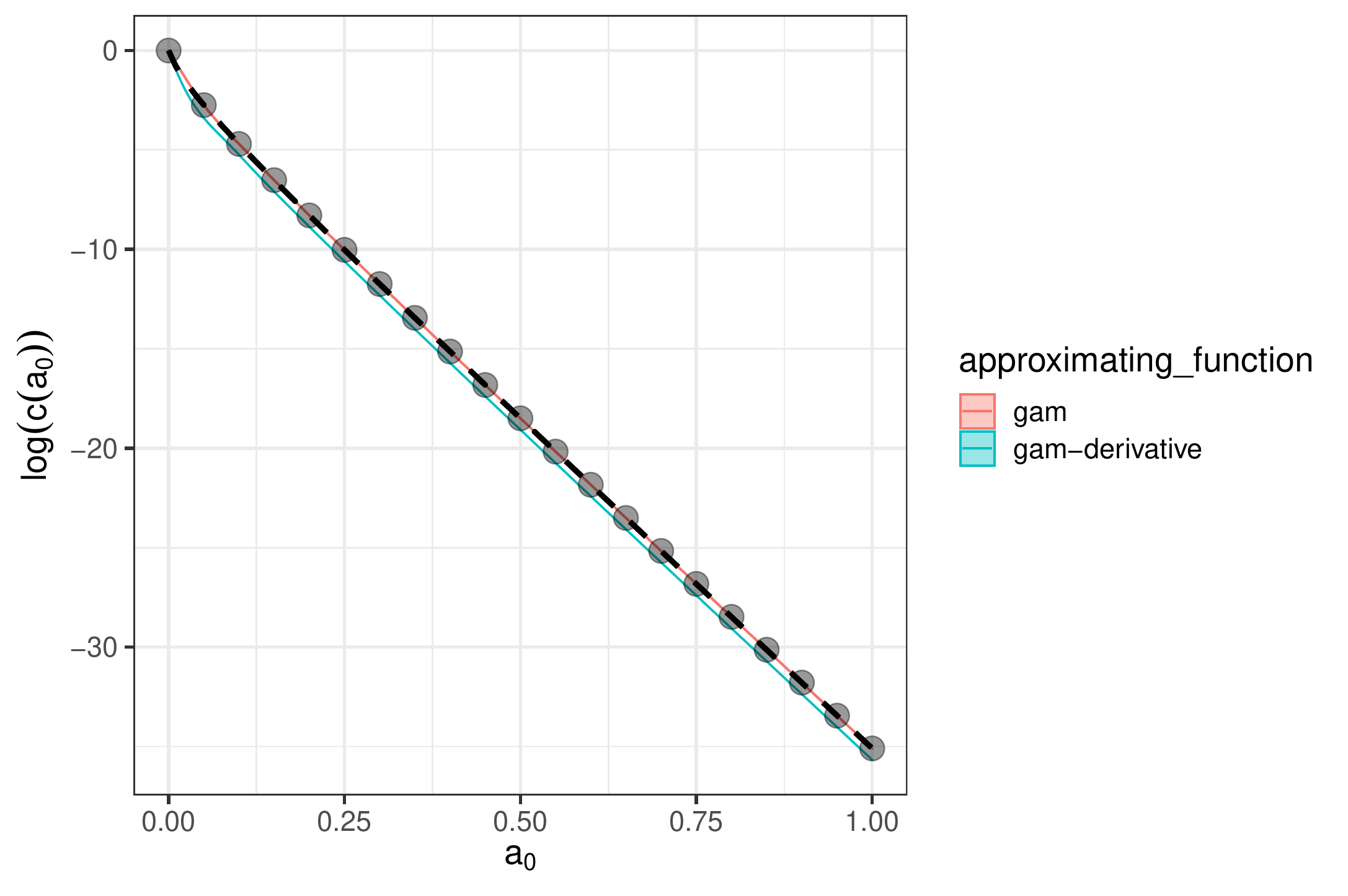}}\\
\hfill
\subfigure[Scenario 3]{\includegraphics[width=7cm]{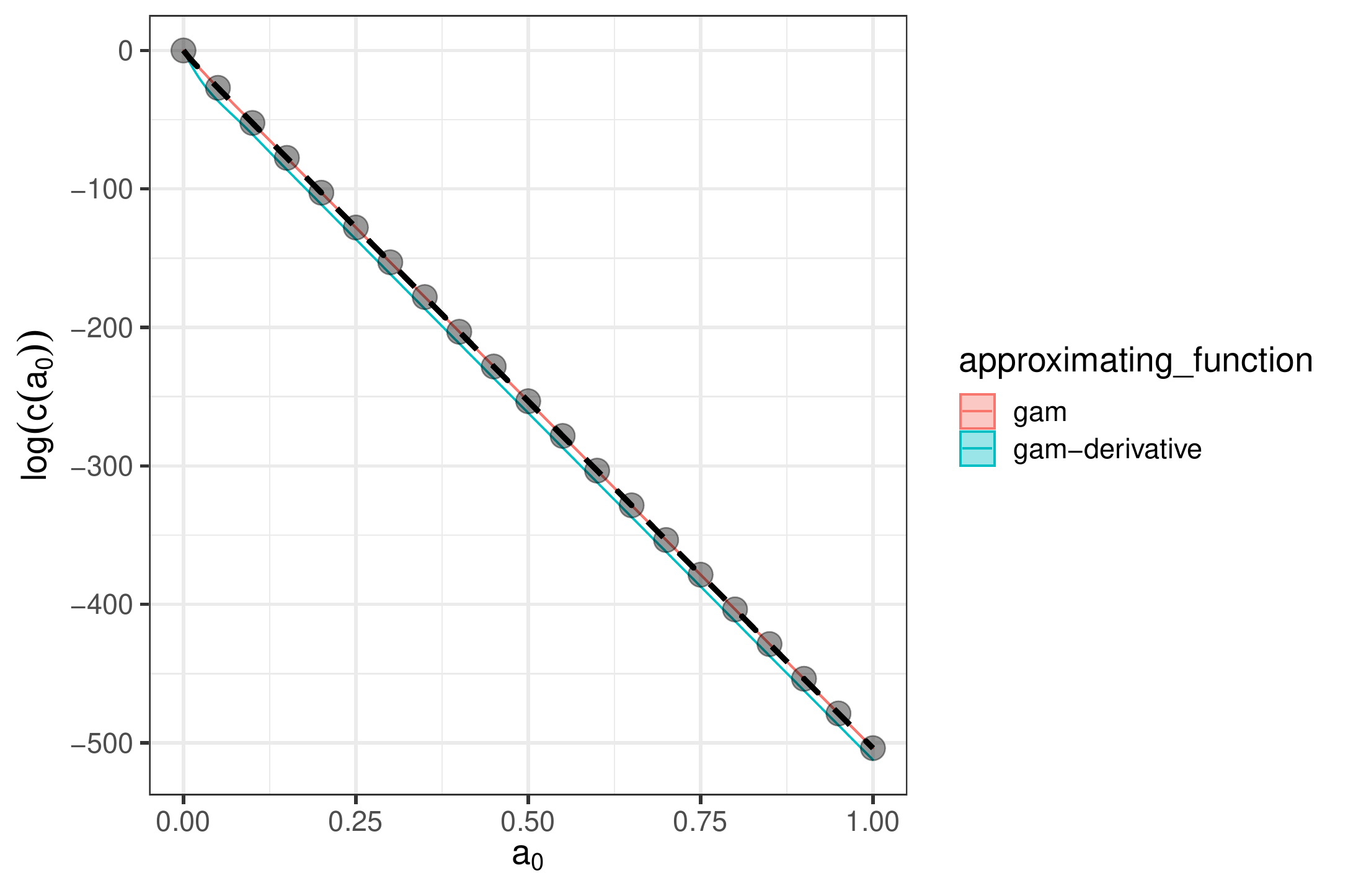}}
\hfill
\subfigure[Scenario 4]{\includegraphics[width=7cm]{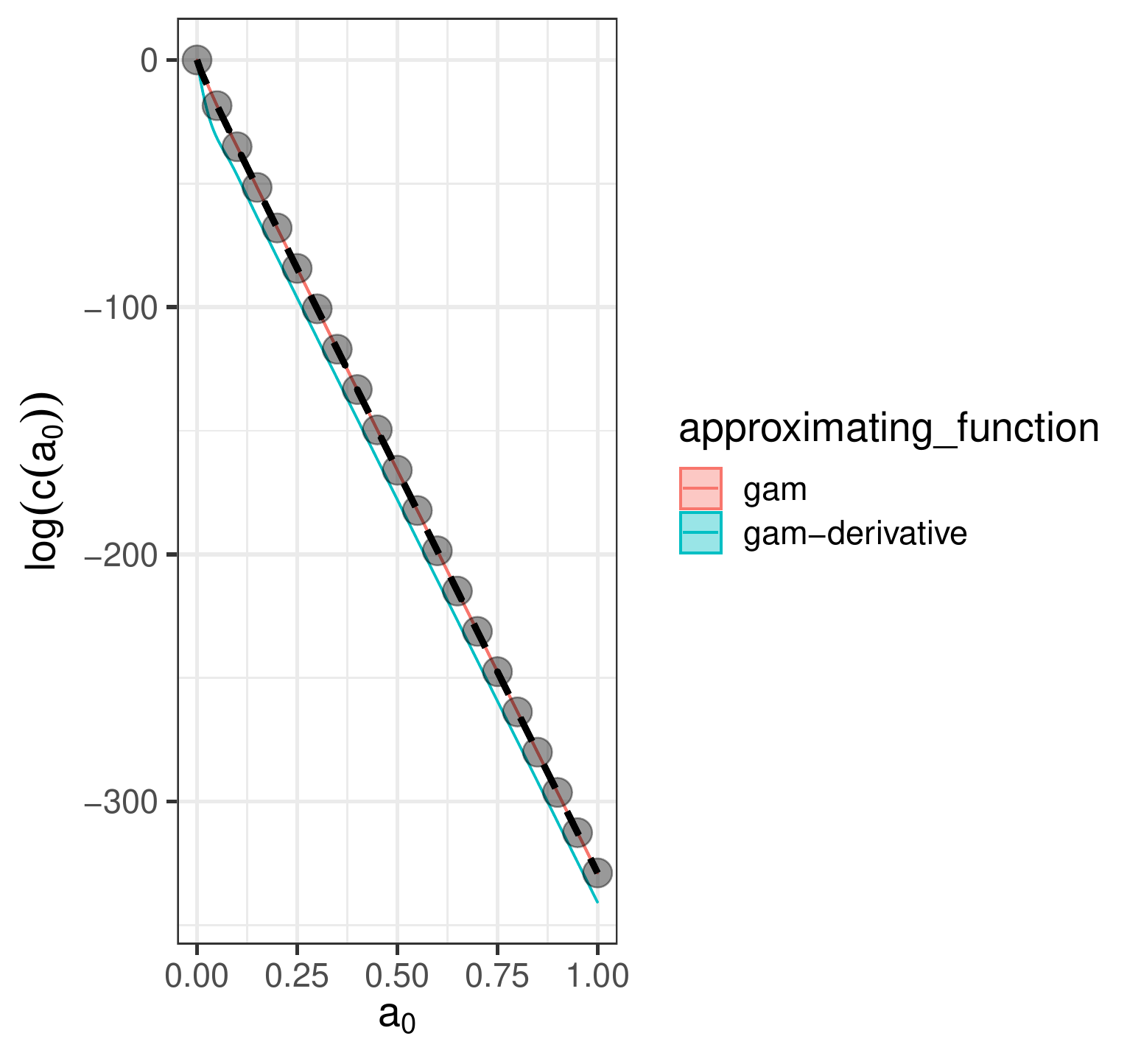}}
\hfill
\end{center}
\caption{\textbf{The log-normalising constant $l(a_0) = \log(c(a_0))$ for the Bernoulli example in each scenario}.
We show the true value of the function (black dashed line) along with the GAM-based approximation (``gam'') and an approximation based on fitting a GAM to the estimated values of $l^\prime(a_0)$ and then using midpoint integration to get $l(a_0)$ (``gam-derivative'').
Colours show the approximation method used.
All results are shown for computations using $J = 20$ points (see Section~\ref{sec:adapt_grid} in the main text).
Please note that y-axes differ between panels.
}
\label{sfig:ca0_Bernoulli}
\end{figure}

\begin{figure}[!ht]
\begin{center}
 \hfill
\subfigure[Scenario A]{\includegraphics[width=7cm]{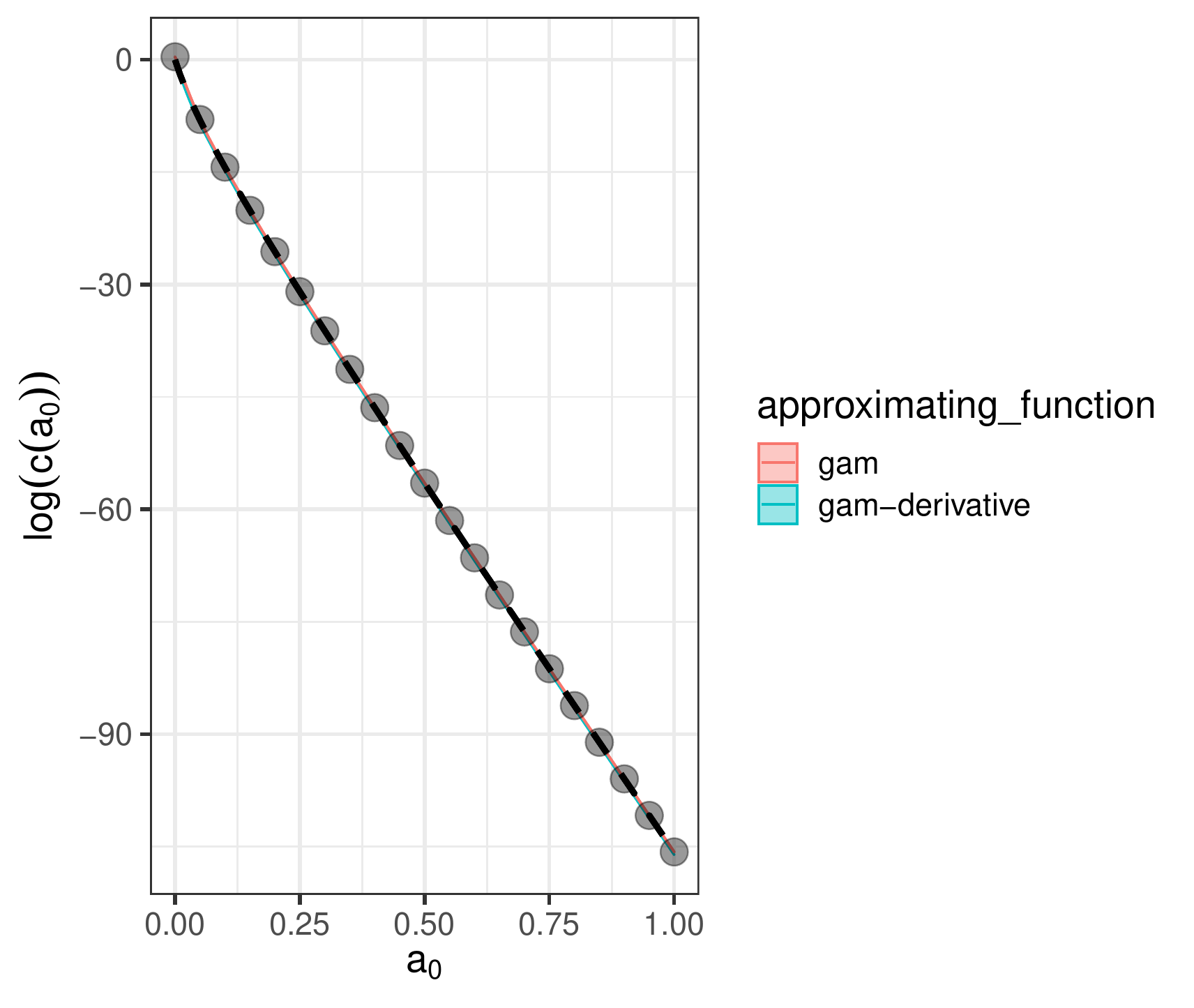}}
\hfill
\subfigure[Scenario B]{\includegraphics[width=7cm]{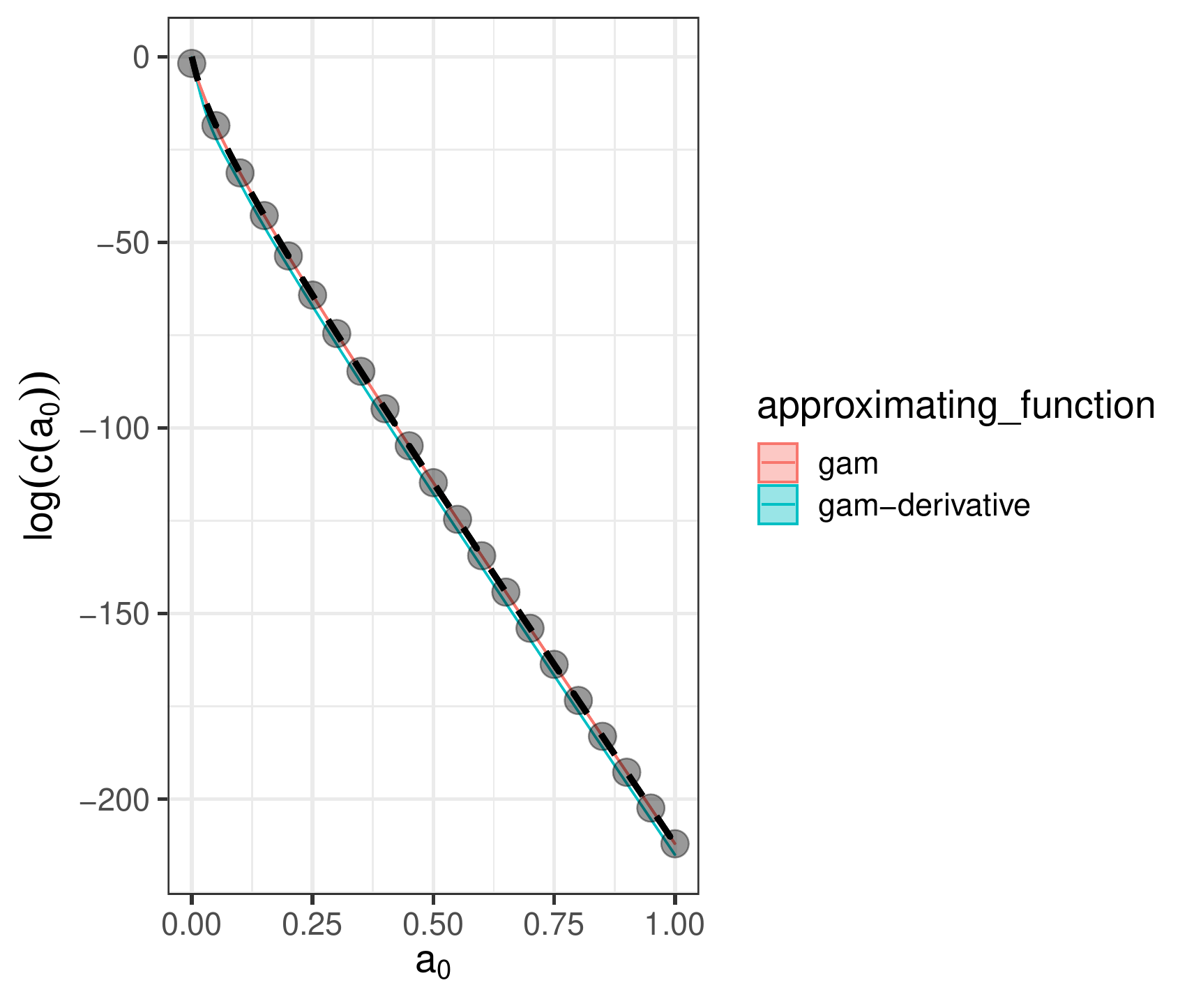}}\\
\hfill
\subfigure[Scenario C]{\includegraphics[width=7cm]{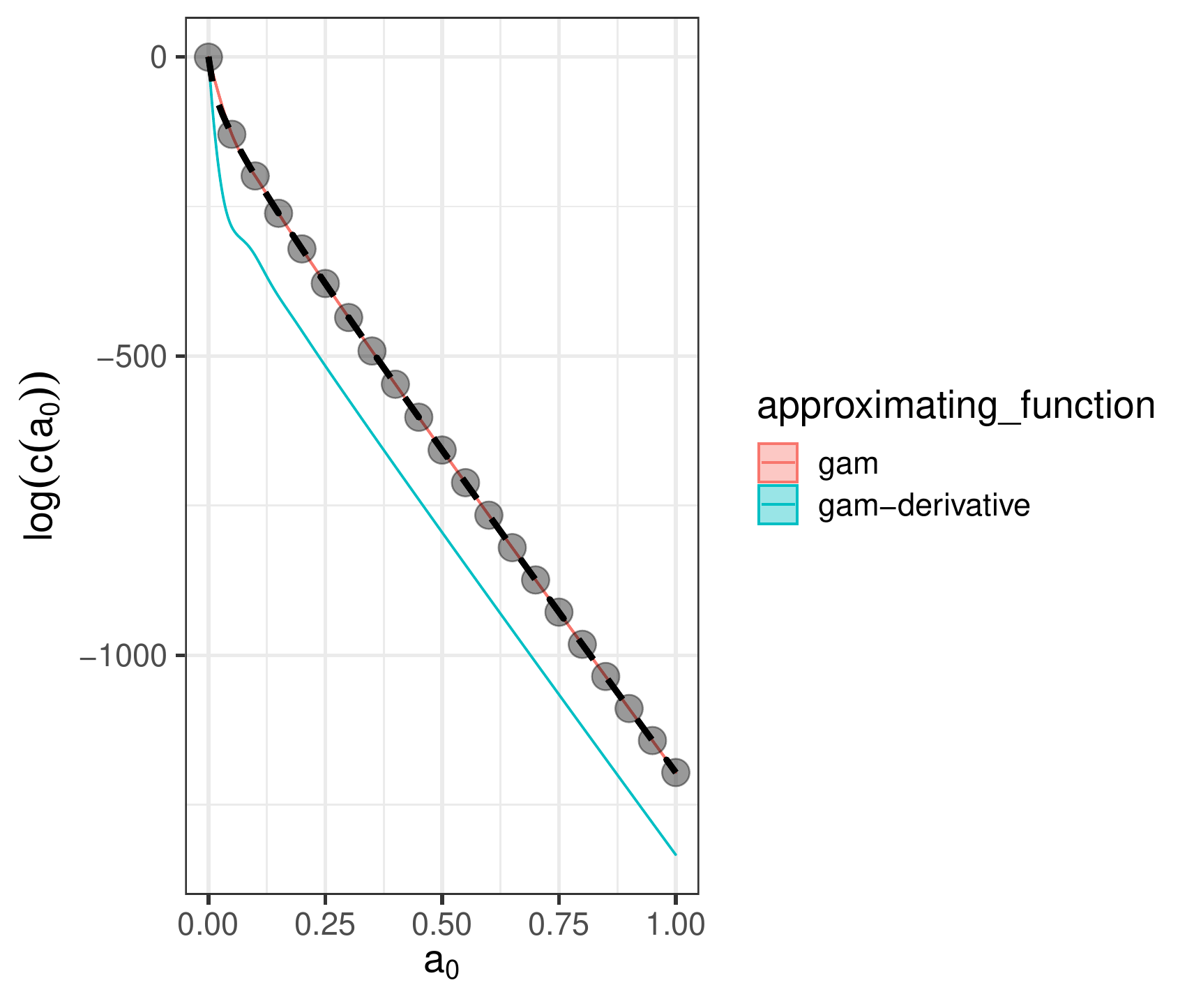}}
\hfill
\subfigure[Scenario D]{\includegraphics[width=7cm]{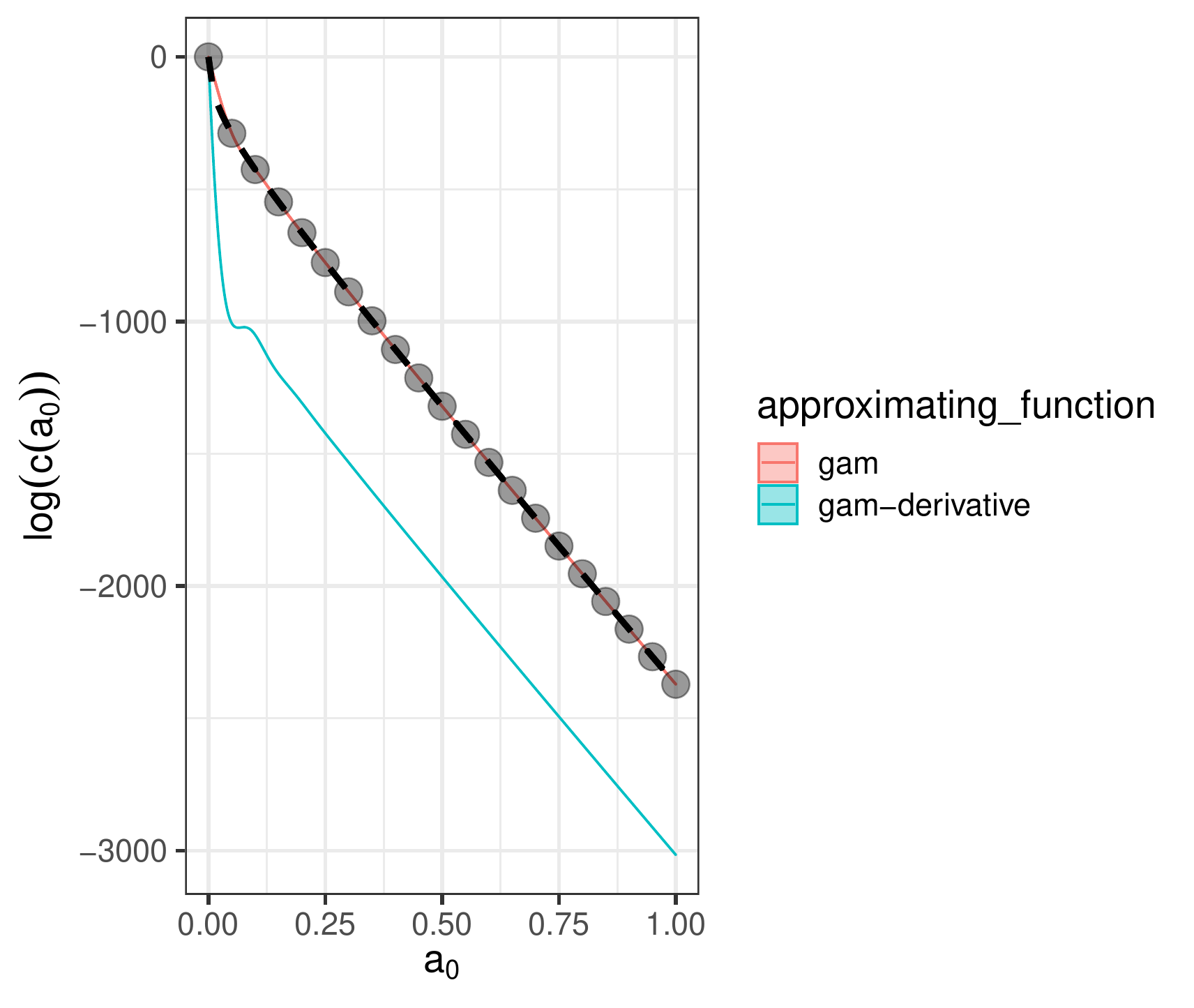}}
\hfill
\end{center}
\caption{\textbf{The log-normalising constant $l(a_0) = \log(c(a_0))$ for the linear regression example in each scenario}.
See Table~\ref{tab:results_NIGregression_scenarios} for details on the configurations of each scenario.
We show the true value of the function (black dashed line) along with the GAM-based approximation (``gam'') and an approximation based on fitting a GAM to the estimated values of $l^\prime(a_0)$ and then using midpoint integration to get $l(a_0)$ (``gam-derivative'').
Colours show the approximation method used.
All results are shown for computations using $J = 20$ points (see Section~\ref{sec:adapt_grid} in the main text).
Please note that y-axes differ between panels.
}
\label{sfig:ca0_NIGRegression}
\end{figure}

\begin{figure}[!ht]
\begin{center}
 \hfill
\subfigure[Scenario A]{\includegraphics[width=7cm]{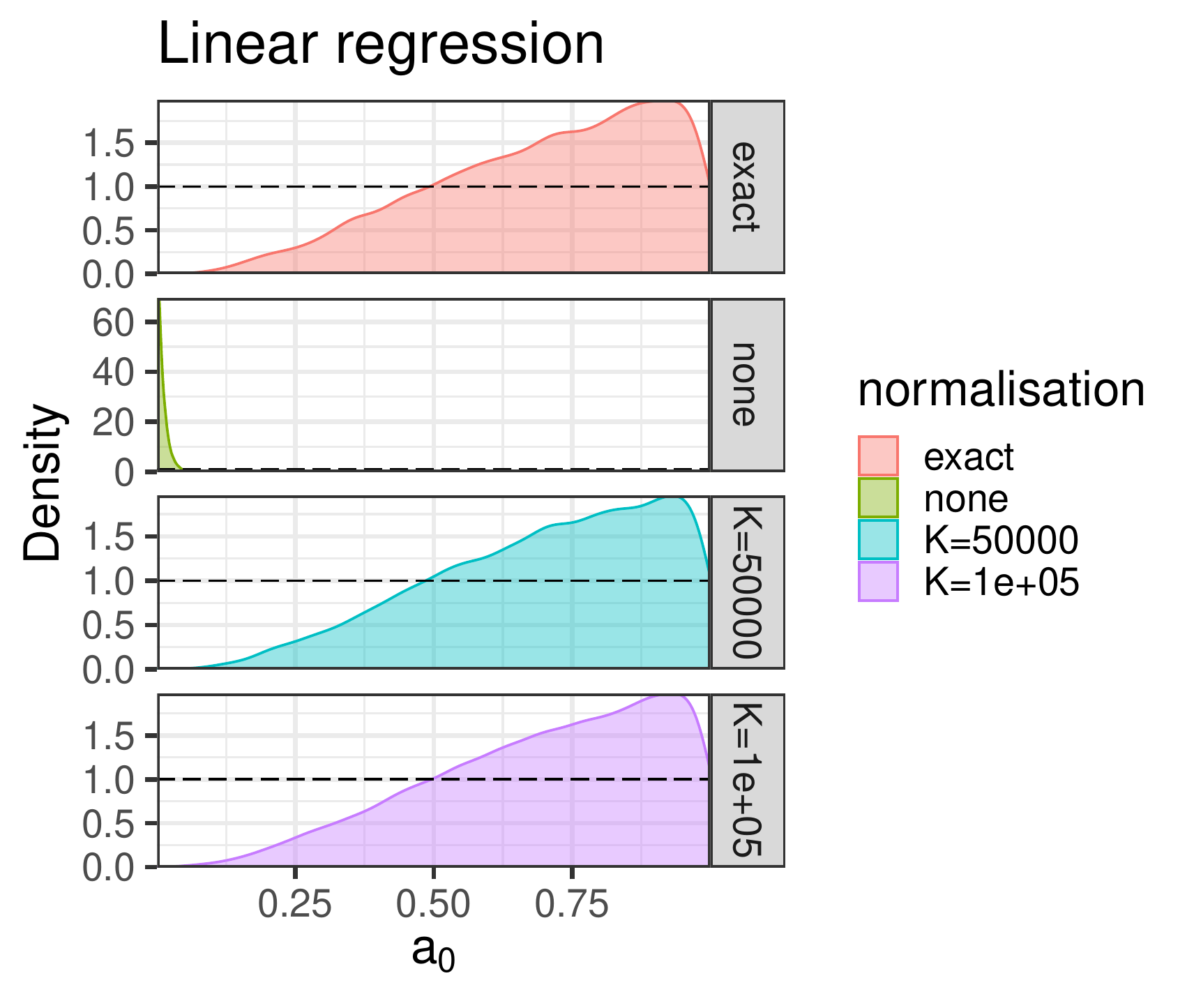}}
\hfill
\subfigure[Scenario B]{\includegraphics[width=7cm]{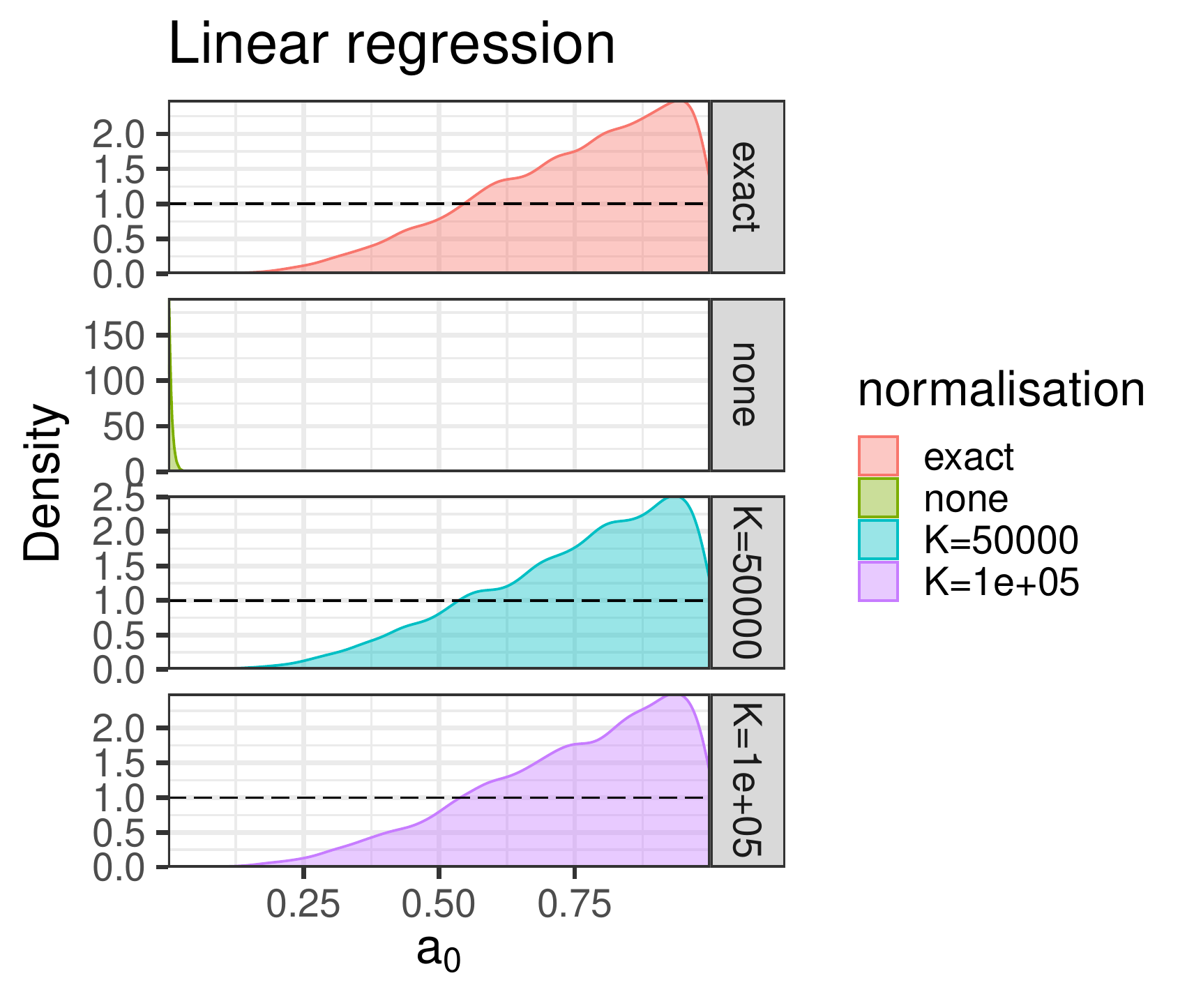}}\\
\hfill
\subfigure[Scenario C]{\includegraphics[width=7cm]{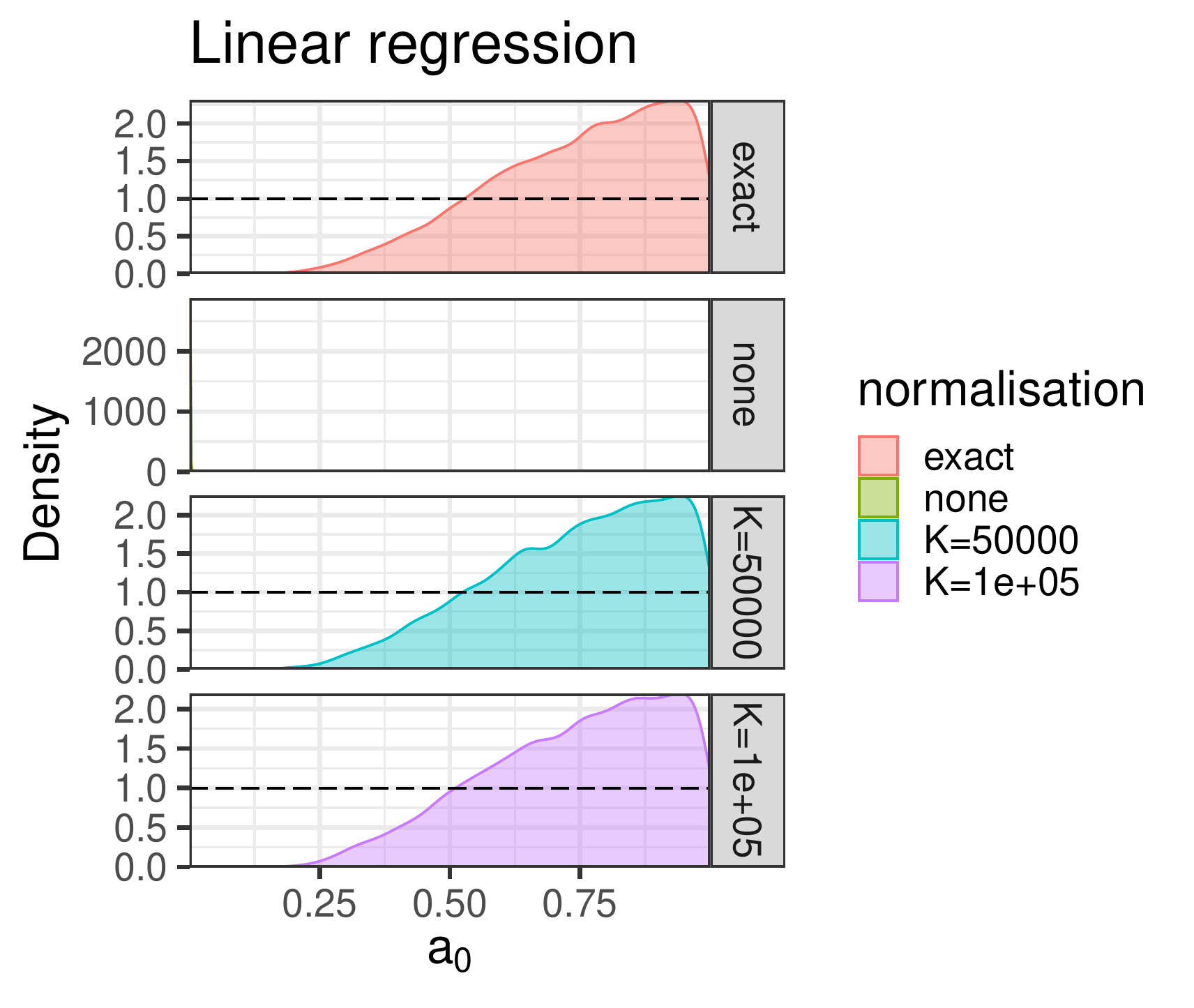}}
\hfill
\subfigure[Scenario D]{\includegraphics[width=7cm]{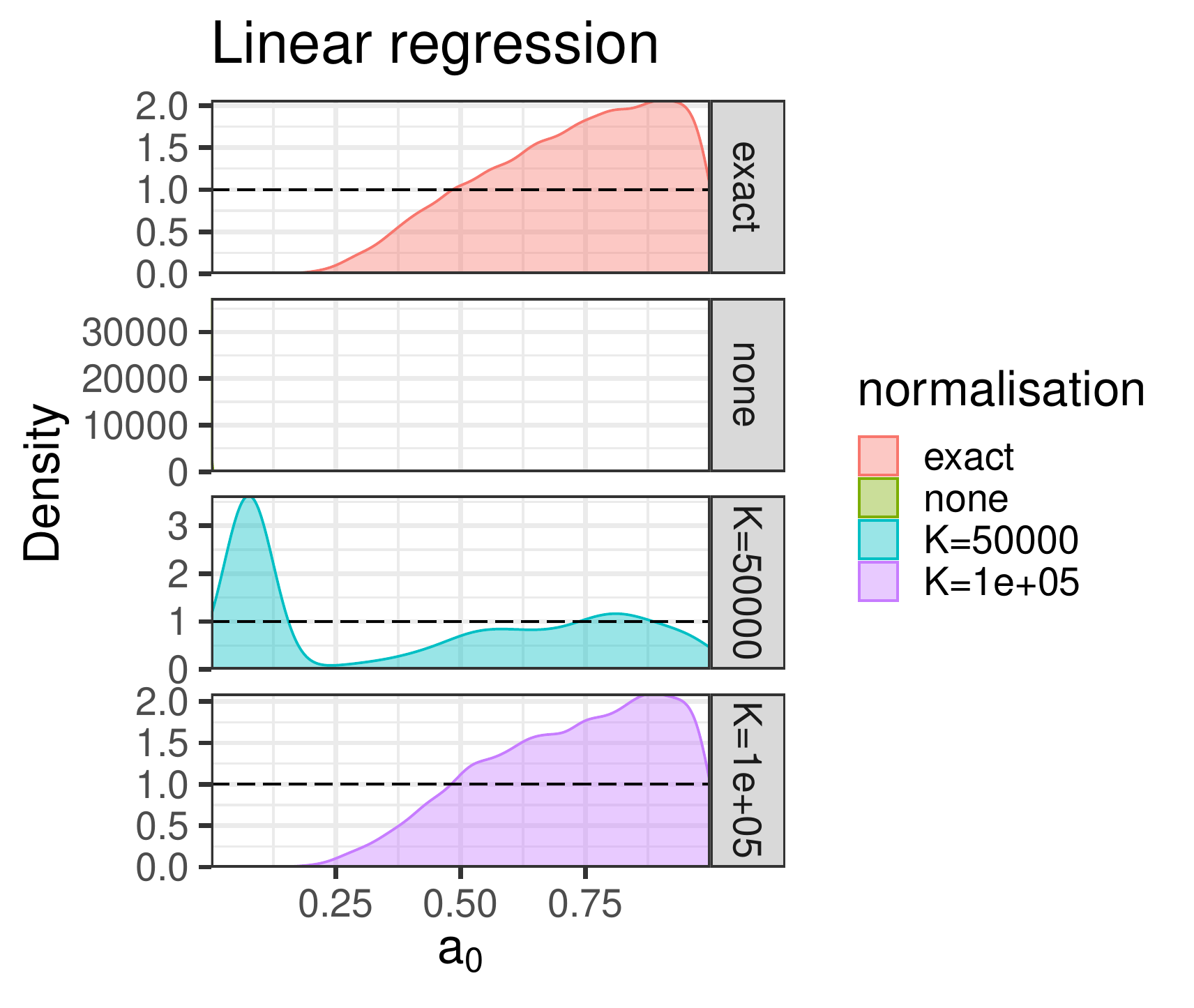}}
\hfill
\end{center}
\caption{\textbf{Posterior distribution for $a_0$  in the linear regression example in each scenario}.
See Table~\ref{tab:results_NIGregression_scenarios} for details on the configurations of each scenario.
We show the unnormalised, exactly normalised and approximately normalised posteriors for grid sizes $K = 50, 000$ and $K = 100, 000$.
The dashed line shows the prior for $a_0$, $\operatorname{Beta}(\eta = 1, \nu = 1)$.
}
\label{sfig:a0_posterior_NIGRegression_scenarios}
\end{figure}

\begin{figure}[!ht]
\begin{center}
 \hfill
\subfigure[$\alpha = 1.2$, $\boldsymbol\beta = \{ -1, 1, 0.5, -0.5\}$]{\includegraphics[width=7cm]{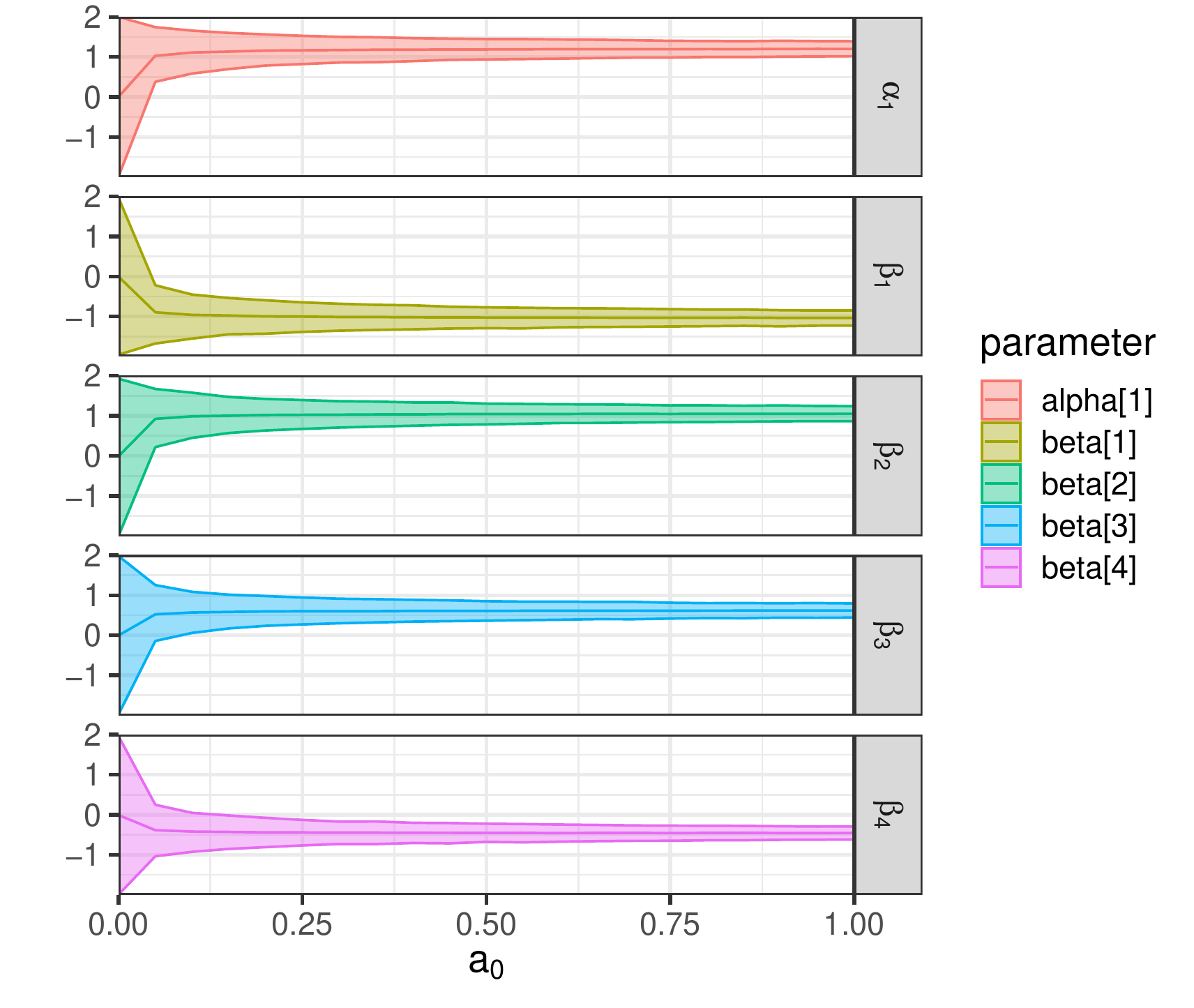}}
\hfill
\subfigure[$\alpha = 0.2$, $\boldsymbol\beta = \{ -10, 1, 5, -5\}$]{\includegraphics[width=7cm]{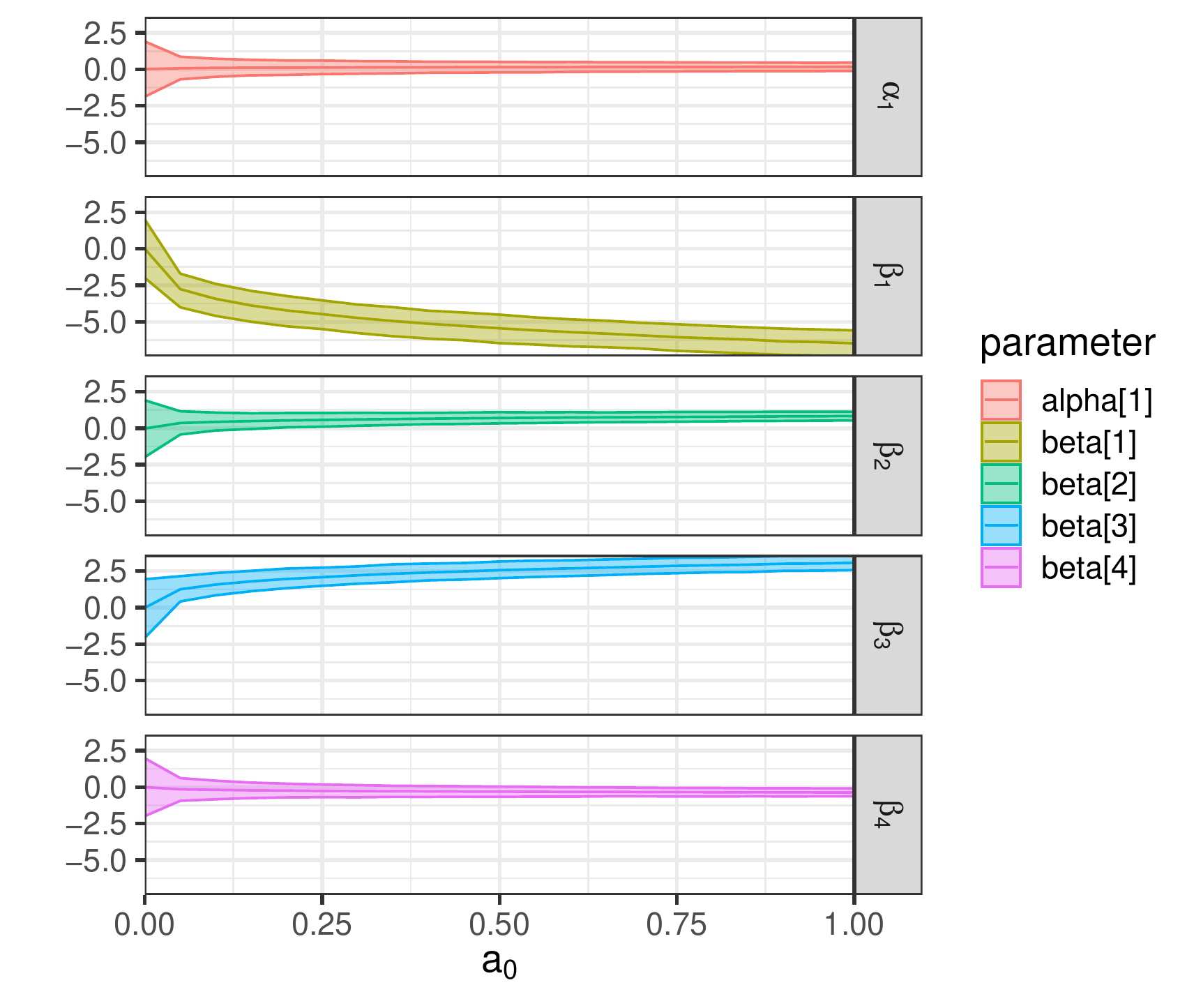}}
\end{center}
\caption{\textbf{Prior sensitivity analysis for the logistic regression example. 
}
We show the means and 95\% credibility intervals for several (fixed) values of $a_0$  for the coefficients and intercept under two sets of data-generating parameters.
}
\label{sfig:sensitivity_logistic_regression}
\end{figure}

\begin{figure}[!ht]
\hfill
\subfigure[Coefficients]{\includegraphics[width=7cm]{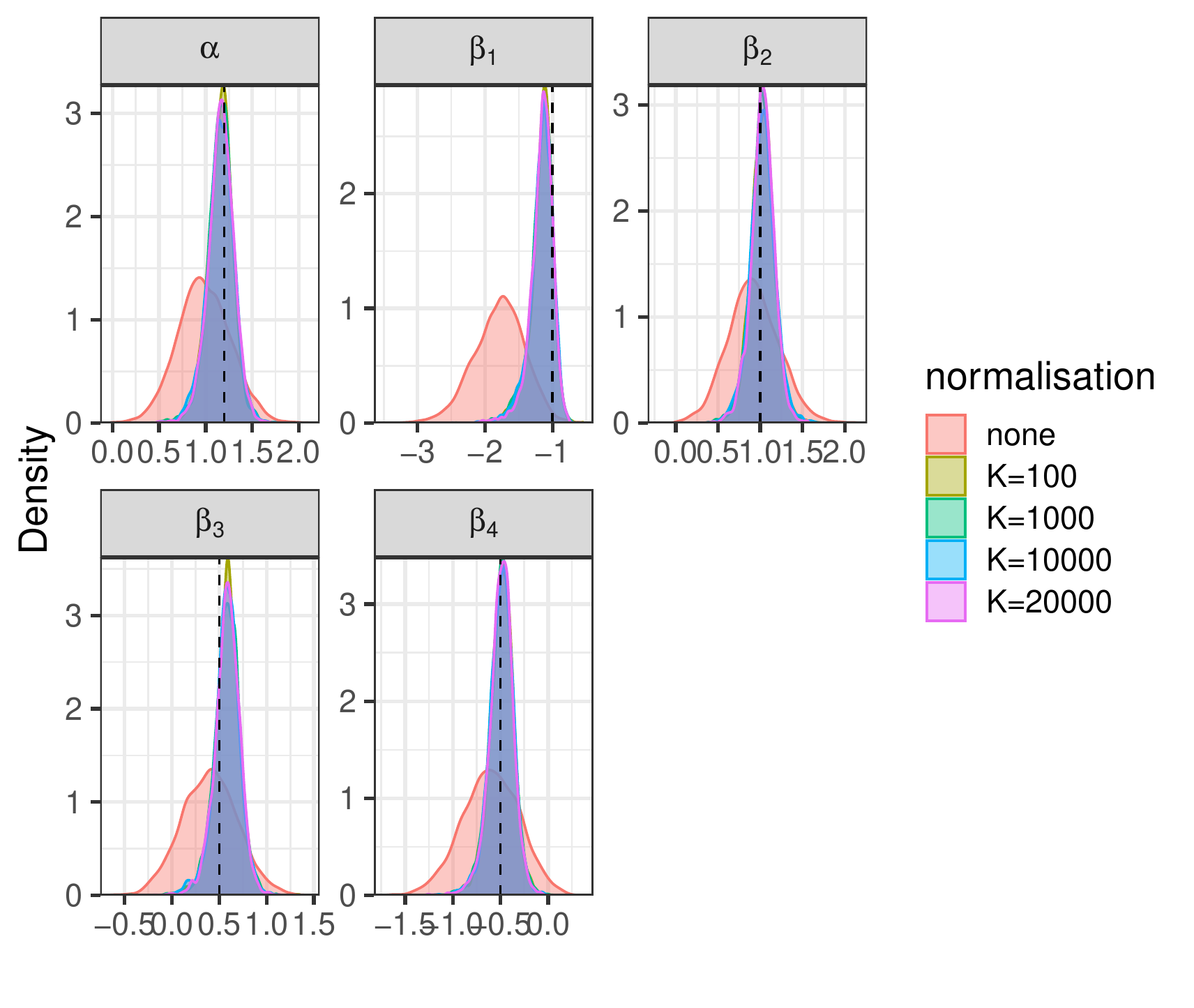}}
\hfill
\subfigure[$a_0$]{\includegraphics[width=7cm]{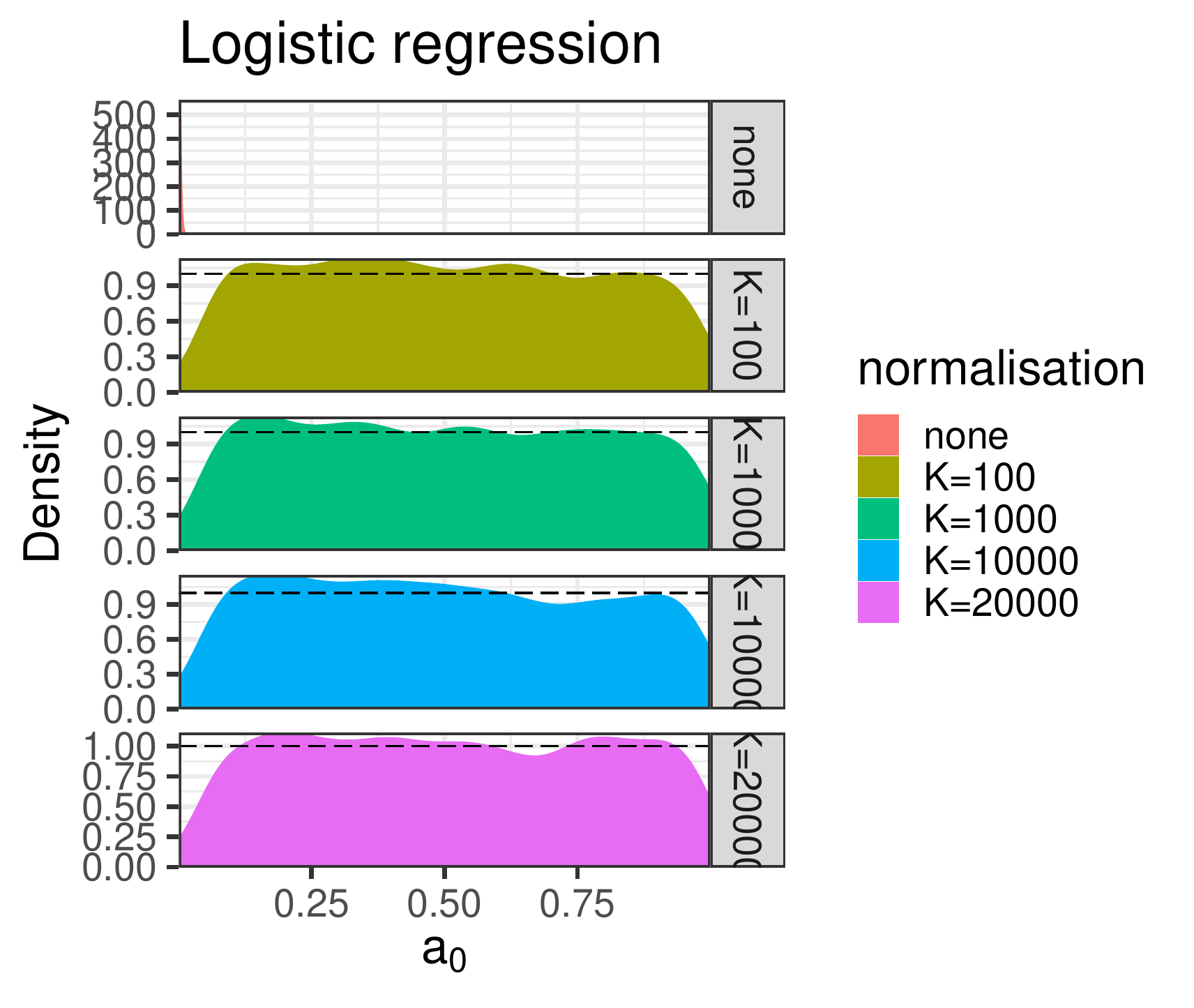}}
\hfill
\caption{\textbf{Results for the logistic regression example with a different set of data-generating parameters}.
In this example, $\alpha = 0.2$, $\boldsymbol\beta = \{ -10, 1, 5, -5\}$.
Panel (a) shows the marginal posterior distributions for model parameters, with colours again pertaining to the approximation scheme.
Vertical dashed lines show the ``true'' parameter values of the data-generating process.
Horizontal dashed lines show the prior density of a $\operatorname{Beta}(\eta = 1, \nu = 1)$ for $a_0$.
In panel (b) the subpanels (and colours) correspond to the posterior distribution of the parameter $a_0$ when $c(a_0)$ is accounted for using various grid sizes $K$ and when it is not included.
}
\label{sfig:logistic_regression_extra}
\end{figure}

\begin{figure}[!ht]
\begin{center}
\includegraphics[scale=0.5]{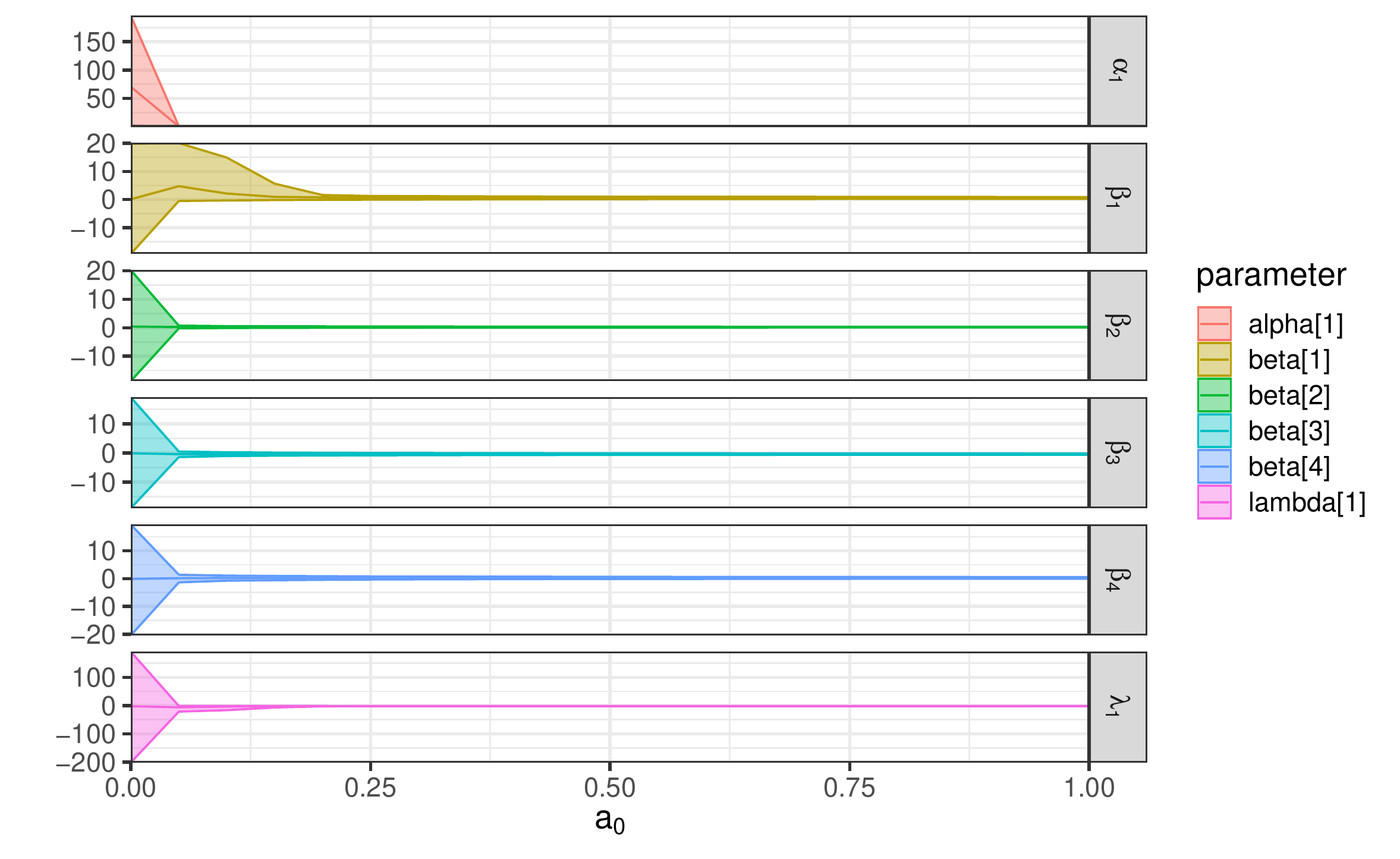} 
\end{center}
\caption{\textbf{Prior sensitivity analysis for the cure rate model}.
Solid line shows the posterior mean, while shaded ribbons show the 95\% credibility intervals.
}
\label{sfig:cure_rate_sensitivity}
\end{figure}

\end{document}